\documentclass{article}
\usepackage{natbib}
\usepackage{amsmath,amssymb,amsfonts}
\usepackage[boxed,lined,linesnumbered]{algorithm2e}
\usepackage{graphicx}
\usepackage{wrapfig}
\usepackage{subcaption}
\usepackage{algorithmic}
\usepackage{tikz-cd}
\usetikzlibrary{arrows,matrix}
\usetikzlibrary{decorations.pathreplacing,angles,quotes}
\newtheorem{thm}{Theorem}
\newtheorem{defn}[thm]{Definition}
\newtheorem{lemma}[thm]{Lemma}
\newtheorem{lem}[thm]{Lemma}
\newtheorem{prop}[thm]{Proposition}
\newtheorem{cor}[thm]{Corollary}

\usepackage[margin=3cm, top=1.5in]{geometry}

\newcommand{\floor}[1]{\lfloor #1 \rfloor}

\newcommand{\halmos}{\hspace*{\fill}\rule{1ex}{1.4ex}}
\makeatletter
\def\newproof#1{\@nprf{#1}}
\def\@nprf#1#2{\expandafter\@ifdefinable\csname #1\endcsname
\global\@namedef{#1}{\@prf{#1}{#2}}\global\@namedef{end#1}{\@endproof}}
\def\@prf#1#2{\@beginproof{#2}{\csname the#1\endcsname}\ignorespaces}
\def\@beginproof#1{\rm \trivlist \item[\hskip \labelsep{\bf #1: }]}
\def\@endproof{\halmos \endtrivlist}
\makeatother
\newproof{proof}{Proof}

\begin{document}

\title{Local Differential Privacy for Physical Sensor Data and Sparse Recovery\\
}
\date{}

\author{Anna C. Gilbert\\ Department of Mathematics\\ University of Michigan \and Audra McMillan \\Department of Mathematics\\University of Michigan}

\maketitle

\begin{abstract}
In this work, we exploit the ill-posedness of linear inverse problems to design algorithms to release differentially private data or measurements of the physical system. We discuss the spectral requirements on a matrix such that only a small amount of noise is needed to achieve privacy and contrast this with the ill-conditionedness. We then instantiate our framework with several diffusion operators and explore recovery via $\ell_1$ constrained minimisation. Our work indicates that it is possible to produce locally private sensor measurements that both keep the exact locations of the heat sources private and permit recovery of the ``general geographic vicinity'' of the sources. 
\end{abstract}


\section{Introduction}

Imagine dropping a few drops of ink into a glass of water. The ink drops spread out, forming complicated tendrils that coil back on each other, expanding quickly, until all of the ink has diffused and the liquid is a slightly darker shade than its original colour. 
There is no physical process by which you can make the diffusing ink coalesce \emph{back} into its original droplets. This intuition is at the heart of what we call \textbf{computational cloaking}. Because it is physically impossible to reconstruct the ink droplet exactly, we should be able to hide or keep private in a precise sense its original location. 
When mathematicians and physicists refer to cloaking, they usually mean transformation optics~\citep{Greenleaf2009}, the design of optical devices with special customised effects on wave propagation. In this paper, we exploit the ill-conditionedness of inverse problems to design algorithms to release differentially private measurements of the physical system.

We are motivated by the explosion in the power and ubiquity of lightweight (thermal, light, motion, etc.) sensors. These data offer important benefits to society. For example, thermal sensor data now plays an important role in controlling HVAC systems and minimising energy consumption in smart buildings \citep{Lin:2002, Beltran:2013}. However, these sensors also collect data inside intimate spaces, homes and workspaces, so the information contained in the data is sensitive.
To continue with the example of thermal sensor data, one might consider sources of heat to be people, whose locations we aim to keep private. 

Our work indicates that it is possible to produce locally differentially private sensor measurements that both keep the exact locations of the heat sources private and permit recovery of the \emph{general vicinity} of the sources. That is, the locally private data can be used to recover an estimate, $\hat{f}$, that is close to the true source locations, $f_0$, in the Earth Mover Distance (EMD). This is the second aspect to our work: algorithms that reconstruct sparse signals with error guarantees with respect to EMD (rather than the more traditional $\ell_1$ or $\ell_2$ error in which accurate recovery is insurmountable).

\subsection{Source Localization}
Suppose that we have a vector $f_0$ of length $n$ that represents the strengths and positions of our ``sources.'' The $i$th entry represents the strength of the source at position $i$. Further, suppose that we take $m$ \textbf{linear} measurements of our source vector; we observe 
\[
	y = M f_0
\]
where $M$ represents some generic linear physical transformation of our original data. Let us also assume that the source vector $f_0$ consists of at most $k$ sources (or $k$ non-zero entries). The straightforward linear inverse problem is to determine $f_0$, given $M$ and a noisy version of $y$. More precisely, given noisy measurements $\tilde y = M f_0 + N(0, \sigma^2 I_m)$, can we produce an estimate $\hat f$ that is still \emph{useful}?

For physical processes such as diffusion, intuitively, we can recover the approximate \emph{geographic vicinity} of the source. This is exactly the concept of \emph{closeness} captured by the Earth Mover Distance (EMD). Thus, in this paper, we aim to recover $\hat{f}$ that is close to $f_0$ in the EMD. The EMD can be defined between any two probability distributions on a finite discrete metric space $(\Omega, d(\cdot, \cdot))$. It computes the amount of \emph{work} required to transform one distribution into the other.

\begin{defn}\citep{Rubner:2000}\label{EMD}
Let $P=\{(x_1, p_1), \cdots, (x_n, p_n)\}$ and $Q=\{(x_1, q_1), \cdots, (x_n, q_n)\}$ be two probability distributions on the discrete space $\{x_1, \cdots, x_n\}$. Now, let 
\begin{equation}\label{EMD}
\mathfrak{f}^* = \arg\min_{\mathfrak{f}\in[0,1]^{n\times n}} \sum_{i=1}^n\sum_{j=1}^n \mathfrak{f}_{ij}d(x_i,x_j)
\end{equation}
\[\text{s.t. } \;\; \mathfrak{f}_{ij}\ge 0 \;\; \forall i,j\in[m],
\;\; \sum_{j=1}^n \mathfrak{f}_{ij}\le p_i \;\; \forall i\in [m],\]
\[ \sum_{i=1}^n \mathfrak{f}_{ij}\le q_i \;\; \forall i\in[n], 
\;\; \text{and}
\;\; \sum_{i=1}^n\sum_{i=1}^n \mathfrak{f}_{ij} = 1.\]
then EMD$(P, Q)= \sum_{i=1}^n\sum_{j=1}^n \mathfrak{f}^*_{ij}d(x_i,x_j).$
\end{defn}


\subsection{Differential Privacy}
To understand our definition of cloaking, we give a very brief introduction to differential privacy in this section. A more in-depth introduction can be found in \cite{Dwork:2014}.
Differential privacy has emerged over the past decade as the leading definition of privacy for privacy-preserving data analysis. A database is a vector $D$ in $\mathcal{D}^n$ for some data universe $\mathcal{D}$. We call two databases $D$, $D'$ adjacent or ``neighbouring'' if $\|D-D'\|_0=1$. 

\begin{defn}[$(\epsilon, \delta)$-Differential Privacy]\citep{Dwork:2006}
A randomised algorithm $\mathcal{A}$ is $(\epsilon, \delta)$-differentially private if for all adjacent databases $D$, $D'$ and events $E$, 
\[
	\mathbb{P}(\mathcal{A}(D)\in E)\le e^{\epsilon}\mathbb{P}(\mathcal{A}(D')\in E)+\delta.
\]
\end{defn}

To understand this definition suppose the database $D$ contains some sensitive information about Charlie and the data analyst, Lucy, produces some statistic $\mathcal{A}(D)$ about the database $D$ via a differentially private algorithm. Then Lucy can give Charlie the following guarantee: an adversary given access to the output $\mathcal{A}(D)$ can not determine whether the database was $D$ or $D'$, where $D$ has Charlie's true data and $D'$ has Charlie's data replaced with an arbitrary element of $\mathcal{D}$. 

\subsection{Computational Cloaking Precisely}

First, we clarify exactly what \emph{information} we would like to keep private. We consider the coordinates of $f_0$ to be our data, that is the locations of the sources are what we would like to keep private. We assume that there exists a metric $d(\cdot, \cdot)$ on the set of possible source locations, which induces the EMD on the set of source vectors. For the remainder of this work, we will assume that the metric $d$ is such that every pair of source locations is connected by a path that travels via neighbours.

When the matrix $M$ represents a physical process, we usually cannot hope to keep the \emph{existence} of a source private and also recover an estimation to $f_0$ that is close in the EMD. However, it may be possible to keep the exact location private while allowing recovery of the ``general vicinity" of the source. In fact, we will show in Section \ref{heatonline} that this is possible for diffusion on the discrete 1-dimensional line and in Section~\ref{sec:graphdiffusion} that we can generalise these results to diffusion on a general graph. We are going narrow our definition of ``neighbouring'' databases to capture this idea.

\begin{defn}\label{neighbours}  For $\alpha>0$, two source vectors $f_0$ and $f_0'$ are $\alpha$-\emph{neighbours} if \[\text{EMD}(f_0, f_0')\le \alpha.\]
\end{defn}
The larger $\alpha$ is, the less stringent the neighbouring condition is, so the more privacy we are providing. This definition has two important instances. 
We can move a source of weight 1 by $\alpha$ units, hiding the location of a large heat source (like a fire) within a small area. Also, we can move a source with weight $\alpha$ by 1 unit, hiding the location that small heat source (like a person) over a much larger area.
We will usually drop the $\alpha$ when referring to neighbouring vectors.

A \emph{locally differentially private} algorithm is a private algorithm in which the individual data points are made private before they are collated by the data analyst. In many of our motivating examples the measurements $y_i$ are at distinct locations prior to being transmitted by a data analyst (for example, at the sensors). Thus, the ``local" part of the title refers to the fact that we consider algorithms where each measurement, $y_i$, is made private individually. This is desirable since the data analyst (e.g. landlord, government) is often the entity the consumer would like to be protected against. Also, it is often the case that the data must be communicated via some untrusted channel \citep{Walters:2007, FTC:2015}.  Usually this step would involve encrypting the data, incurring significant computational and communication overhead. However, if the data is made private prior to being sent, then there is less need for encryption. We then wish to use this locally differentially private data to recover an estimate to the source vector that is close in the EMD. The structure of the problem is outlined in the following diagram:

{
\begin{center}
\begin{tikzpicture}
\node at (-0.2,2) {$f_0$};
\draw [->] (0.1,2) -- node[above]{$M$} (0.8,2);
\draw [decoration={brace, raise=5pt},decorate] (1.2, 1.5) -- (1.2, 2.5);
\node at (1.5, 1.5) {$y_m$};
\node at (1.5, 2.5) {$y_1$};
\draw [dotted] (1.5, 1.7) -- (1.5, 2.3);
\draw [->] (1.8, 1.5) -- node[above]{$\mathcal{A}$} (2.6,1.5);
\draw [->] (1.8, 2.5) -- node[above]{$\mathcal{A}$} (2.6, 2.5);
\node at (2.9, 1.5) {$\tilde{y_m}$};
\node at (2.9, 2.5) {$\tilde{y_1}$};
\draw [dotted] (2.9, 1.8) -- (2.9, 2.3);
\draw [decoration={brace, mirror, raise=5pt},decorate] (3.1, 1.5) -- (3.1, 2.5);
\draw [->] (3.5,2) -- node[above]{$R$} (4.1,2);
\node at (4.3, 2) {$\hat{f}$};
\end{tikzpicture}
\end{center}
}

Design algorithms $\mathcal{A}$ and $R$ such that:
\begin{enumerate}
\item (Privacy) For all neighbouring source vectors $f_0$ and $f_0'$, indices $i$, and Borel measurable sets $E$ we have $$\mathbb{P}(\mathcal{A}((Mf_0)_i)\in E)\le e^{\epsilon}\mathbb{P}(\mathcal{A}(M{f_0'})_i\in E)+\delta.$$
\item (Utility) $\text{EMD}(f_0, \hat{f})$ is small.
\end{enumerate}

\subsection{Related Work}
An in-depth survey on differential privacy and its links to machine learning and signal processing can be found in \citep{Sarwate:2013}. The body of literature on general and local differential privacy is vast and so we restrict our discussion to work that is directly related. There is a growing body of literature of differentially private sensor data \citep{Liu:2012, Li:2015, Wang:2016, Jelasity:2014, Eibl:2016}. Much of this work is concerned with differentially private release of aggregate statistics derived from sensor data and the difficulty in maintaining privacy over a period of time (called the continual monitoring problem). 

Connections between privacy and signal recovery have been explored previously in the literature. \cite{Dwork:2007} considered the recovery problem with noisy measurements where the matrix $M$ has i.i.d.~standard Gaussian entries. 
Newer results of \cite{Bun:2014} can be interpreted in a similar light where $M$ is a binary matrix. Compressed sensing has also been used in the privacy literature as a way to reduce the amount of noise needed to maintain privacy \citep{Li:2011, Roozgard:2016}.

There are also several connections between sparse signal recovery and inverse problems \citep{Farmer:2013, Burger:2010, Haber:2008, Landa:2011}. The heat source identification problem is severely ill-conditioned and, hence, it is known that noisy recovery is impossible in the common norms like $\ell_1$ and $\ell_2$. This has resulted in a lack of interest in developing theoretical bounds \citep{Li:2014}, thus the mathematical analysis and numerical algorithms for inverse heat source problems are still very limited. 

To the best of the author's knowledge, the papers that are most closely related to this work are \cite{Li:2014}, \cite{Beddiaf:2015} and \cite{Bernstein:2017}. All these papers attempt to circumvent the condition number lower bounds by changing the error metric to capture ``the recovered solution is geographically close to the true solution", as in this paper. Our algorithm is the same as Li, et al., who also consider the Earth Mover Distance (EMD). Our upper bound is a generalisation of theirs to source vectors with more than one source. Beddiaf, et al.~follows a line of work that attempts to find the sources using $\ell_2$-minimisation and regularisation. In work that was concurrent to ours, Bernstein et al. also considered heat source location, framed as deconvolution of the Gaussian kernel. They proved that a slight variant of Basis Pursuit Denoising solves the problem exactly assuming enough sensors and sufficient separability between sources. They also arrive at a similar result to Theorem \ref{main} for the noisy case \cite[Theorem 2.7]{Bernstein:2017}.
\section{Privacy of measurements and Ill-conditioned Matrices}

\subsection{The Private Algorithm}

Because we assume that our sensors are lightweight computationally, the algorithm $\mathcal{A}$ is simply for each sensor to add Gaussian noise locally to its own measurement before sending to the perturbed measurement to the central node\footnote{Gaussian noise is not the only option to achieve privacy. There has been some work on the optimal \emph{type} of noise to add to achieve privacy \citep{Geng:2016}.}. The question then is; how much noise should we add to maintain privacy? The following lemma says, essentially, that the standard deviation of the noise added to a statistic should be proportional to how much the statistic can vary between neighbouring data sets. Let $g:\mathcal{D}^n\to\mathbb{R}^n$ be a function and let $\triangle_2 g = \max_{D, D' \text{neighbours}} \|g(D)-g(D')\|_2$ (called the $\ell_2$ sensitivity of $g$). 

\begin{lemma}[The Gaussian Mechanism]\label{gaussianmechanism} \citep{Dwork:2014} Let $\epsilon>0$, $\delta>0$ and $\sigma = \frac{2\ln(1.25/\delta)\triangle_2 g}{\epsilon}$ then $$\mathcal{A}(D) \sim g(D)+N(0, \sigma^2I_n)$$ is an $(\epsilon, \delta)$-differentially private algorithm.
\end{lemma}

Let us apply the Gaussian mechanism to the general linear inverse problem. As we discussed previously, ill-conditioned source localization problems behave poorly under addition of noise. Intuitively, this should mean we need only add a small amount of noise to mask the original data. We show that this statement is partially true. However, there is a fundamental difference between the notion of a problem being ill-conditioned (as defined by the condition number) and being easily kept private. Let $M_i$ be the $i$th column of $M$.

\begin{prop}\label{propstddev}
With $\alpha>0$ and the definition of $\alpha$-neighbours presented in Definition \ref{neighbours}, we have $$\mathcal{A}(Mf_0) \sim Mf_0+\frac{2\log(1.25/\delta)\triangle_2(M)}{\epsilon}N(0,I_m)$$ is a $(\epsilon, \delta)$-differentially private algorithm where
 \[\triangle_2 (M) = \alpha\max_{e_i,e_j\; \text{neighbours}}\|M_i-M_j\|_2\]
\end{prop}

\begin{proof}
Suppose $f_0$ and $f_0'$ are $\alpha$-neighbours and let $f_{kl}$ be the optimal flow from $f_0$ to $f_0'$ (as defined in Definition \ref{EMD}) so $f_0=\sum_{k,l}f_{kl}e_k$ and $f_0'=\sum_{k,l}f_{kl}e_l$, where $e_k$ are the standard basis vectors. Then 
\begin{align*}
\|Mf_0-&Mf_0'\|_2 \le \sum_{k,l}f_{kl}\|Me_k-Me_l\|_2\\\
&= \left(\sum_{k,l} f_{kl}d(e_k, e_l)\right)\max_{i, j\text{neighbours}}\|Me_j-Me_{i}\|_2\\
&\le \alpha \max_{i, j\text{neighbours}} \|Me_i-Me_{j}\|_2
\end{align*}
Then the fact that the algorithm is $(\epsilon, \delta)$-differentially private follows from Lemma~\ref{gaussianmechanism}. 
\end{proof}

Let $(s_1, \cdots, s_{\min\{n, m\}})$ be the spectrum of $M$, enumerated such that $s_i\le s_{i+1}$. The \emph{condition number}, $\kappa_2(M)$, is a measure of how ill-conditioned this inverse problem is. It is defined as $$\kappa_2(M) := \max_{e,b\in\mathbb{R}^m\backslash\{0\}} \frac{\|M^{+}b\|_2}{\|M^{+}e\|_2}\frac{\|e\|_2}{\|b\|_2} = \frac{s_{\max\{m,n\}}(M)}{s_{1}(M)}$$ where $M^+$ is the pseudo inverse of $M$. The larger the condition number the more ill-conditioned the problem is \cite{Belsley:1980}.  

The following matrix illustrates the difference between how ill-conditioned a matrix is and how much noise we need to add to maintain privacy. Suppose \[M=\begin{pmatrix}
1 & 0\\
0 & \rho
\end{pmatrix}\] where $\rho<1$ is small. While this problem is ill-conditioned, $\kappa_2(M)=1/\rho$ is large, we still need to add considerable noise to the first coordinate of $Mx_0$ to maintain privacy. 

A necessary condition for $\Delta_2(M)$ to be small is that the matrix $M$ is \emph{almost} rank 1, that is, the spectrum should be almost 1-sparse. In contrast the condition that $\kappa_2(M)$ is large is only a condition on the maximum and minimum singular values. The following lemma says that if the amount of noise we need to add, $\triangle_2(M)$, is small then the problem is necessarily ill-conditioned. 



\begin{lemma}\label{illconditioned}
Let $M$ be a matrix such that $\|M\|_2=1$ then $$\triangle_2(M)\ge \frac{\alpha}{\kappa_2(M)},$$ where $\alpha$ is the parameter in Definition \ref{neighbours}.
\end{lemma}
\begin{proof}
Suppose $e_i$ is a neighbouring source to $e_1$ then
\begin{align*}
\frac{1}{\kappa_2(M)} &= \min_{\text{rank}E<\min\{m, n\}} \|M-E\|_2\le \|M-[M_j \;M_2\; M_3 \cdots M_{n}]\|_2\le \|M_1-M_{i}\|_2. 
\end{align*}
Since we could have replaced 1 and $i$ with any pair of neighbours we have \[\frac{1}{\kappa_2(M)}\le \min_{i, j\text{neighbours}}\|M_i-M_j\|_2 =\frac{\Delta_2(M)}{\alpha}.\]
\end{proof}

The following lemma gives a characterization of $\triangle_2(M)$ in terms of the spectrum of $M$. It verifies that the matrix $M$ must be \emph{almost} rank 1, in the sense that the spectrum should be dominated by the largest singular value. 

\begin{lemma}\label{spectrum}
If $\Delta_2(M)\le\nu$, then $|\|M_i\|_2-\|M_{j}\|_2|~\le~\frac{\nu}{\alpha}$ for any pair of neighbouring locations $e_i$ and $e_j$ and $|\sum_{i\neq \min\{m,n\}}s_i|\le~\frac{(n+1)^{3/2}\rho\nu}{\alpha}$, where $\rho$ is the diameter of the space of source locations. 

Conversely, if $|\sum_{i\neq~\min\{m,n\}}s_i|~\le~\frac{\nu}{\alpha}$ and $|\|M_i\|_2~-~\|M_{j}\|_2|~\le~\frac{\nu}{\alpha}$ then $\Delta_2(M)\le~4\nu$.
\end{lemma}
\begin{proof} Let $e_i$ and $e_j$ be neighbouring sources.
Now, assume $\Delta_2(M)\le \nu$ then \[|\|M_i\|_2~-~\|M_{j}\|_2|~\le~\|M_i~-~M_{j}\|_2~\le~\frac{\nu}{\alpha}.\] Suppose wlog that $\max_i\|M_i\|_2=\|M_1\|_2$ and let $M' = [M_1\cdots M_1]$ be the matrix whose columns are all duplicates of the first column of $M$. Recall that the trace norm of a matrix is the sum of its singular values and for any matrix, $\|M\|_{tr}\le \sqrt{\min\{m,n\}}\|M\|_F$ and $\|M\|_2\le\|M\|_F$. Since $M'$ is rank 1, $\|M'\|_{tr} = \|M'\|_2 = s_{\min\{m,n\}}$
, thus, 
\begin{align*}
|\sum_{i=1}^{\min\{m,n\}-1}s_i| &\le |\|M\|_{tr}-\|M'\|_{tr}|+|\|M'\|_{tr}-s_{\min\{m,n\}}|\\
&\le \|M'-M\|_{tr}+|\|M'\|_2-\|M\|_2|\\
&\le (\sqrt{\min\{n,m\}}+1)\|M'-M\|_F\\
&\le(\sqrt{\min\{n,m\}}+1)\rho(n-1)\frac{\nu}{\alpha}
\end{align*} 
Conversely, suppose $|\sum_{i\neq {\min\{m,n\}}}s_i|\le \frac{\nu}{\alpha}$ and $|\|M_i\|_2-\|M_{j}\|_2|\le \frac{\nu}{\alpha}$. Using the SVD we know, $M=\sum s_i U_i\otimes V_i$ where $U_i$ and $V_i$ are the left and right singular values, respectively. Thus, 
\begin{align*}
\|M_i-M_{j}\|_2 &= \|\sum_k s_k(U_k)_iV_k-\sum_k s_k(U_k)_{j}V_k\|_2\\
&\le s_{\min\{m,n\}}|(U_1)_i-(U_1)_{j}|+\frac{\nu}{\alpha}.
\end{align*}
Also, $\frac{\nu}{\alpha}\ge \|M_i\|_2-\|M_{j}\|_2\ge s_{\min\{m,n\}}(U_1)_i-\frac{\nu}{\alpha}-s_{\min\{m,n\}}(U_1)_{j}-\frac{\nu}{\alpha}$ so\\ $|(U_1)_i-(U_1)_{j}|\le 3\frac{\nu}{\alpha s_{\min\{m,n\}}}$.
\end{proof}

\subsection{Backdoor Access via Pseudo-randomness}\label{backdoor}

It has been explored previously in the privacy literature that replacing a random noise generator with cryptographically secure pseudorandom noise generator in an efficient differentially private algorithm creates an algorithm that satisfies a weaker version of privacy, computational differential privacy \citep{Mironov:2009}. While differential privacy is secure against \emph{any} adversary, computational differential privacy is secure against a \emph{computationally bounded} adversary. In the following definition, $\kappa$ is a \emph{security} parameter that controls various quantities in our construction. 

\begin{defn}[Simulation-based Computational Differential Privacy (SIM-CDP)] \citep{Mironov:2009}
A family, $\{\mathcal{M}_{\kappa}\}_{\kappa\in\mathbb{N}}$, of probabilistic algorithms $\mathcal{M}_{\kappa}:\mathcal{D}^n\to\mathcal{R}_{\kappa}$ is $\epsilon_n$-SIM-CDP if there exists a family of $\epsilon_{\kappa}$-differentially private algorithms $\{\mathcal{A}_n\}_{n\in\mathbb{N}}$, such that for every probabilistic polynomial-time adversary $\mathcal{P}$, every polynomial $p(\cdot)$, every sufficiently large $\kappa\in\mathbb{N}$, every dataset $D\in\mathcal{D}^n$ with $n\le p(\kappa)$, and every advice string $z_{\kappa}$ of size at $p(\kappa)$, it holds that, \[|\mathbb{P}[\mathcal{P}_{\kappa}(\mathcal{M}_{\kappa}(D)=1)-\mathbb{P}[\mathcal{P}_{\kappa}(\mathcal{A}_{\kappa}(D)=1)|\le \text{negl}(\kappa).\] That is, $\mathcal{M}_{\kappa}$ and $\mathcal{A}_{\kappa}$ are computationally indistinguishable.
\end{defn}

The transition to pseudo-randomness, of course, has the obvious advantage that pseudo-random noise is easier to generate than truly random noise. In our case, it also has the additional benefit that, given access to the seed value, pseudo-random noise can be removed, allowing us to build a ``backdoor" into the algorithm. Suppose we have a trusted data analyst who wants access to the most accurate measurement data, but does not have the capacity to protect sensitive data from being intercepted in transmission. Suppose also that this party stores the seed value of each sensor and the randomness in our locally private algorithm $\mathcal{A}$ is replaced with pseudo-randomness. Then, the consumers are protected against an eavesdropping computationally bounded adversary, and the trusted party has access to the noiseless \footnote{This data may still be corrupted by sensor noise that was not intentionally injected} measurement data. This solution may be preferable to simply encrypting the data during transmission since there may be untrusted parties who we wish to give access to the private version of the data. 

\begin{cor}[Informal] Replacing the randomness in Proposition \ref{propstddev} with pseudo-randomness produces a  local simulation-based computational differentially private algorithm for the same task. In addition, any trusted party with access to the seed of the random number generator can use the output of the private algorithm to generate the original data.
\end{cor}

\section{Recovery algorithm and Examples}\label{recoveryalgorithm}

We claimed that the private data is both useful and differentially private. In this Chapter we discuss recovering an estimate of $f_0$ from the noisy data $\tilde{y}$. Algorithms for recovering a sparse vector from noisy data have been explored extensively in the compressed sensing literature. However, theoretical results in this area typically assume that the measurement matrix $M$ is sufficiently nice. Diffusion matrices are typically very far from satisfying the \emph{niceness} conditions required for current theoretical results. Nonetheless, in this Chapter we discuss the use of a common sparse recovery algorithm, Basis Pursuit Denoising (BPD), for ill-conditioned matrices. The use of BPD to recover source vectors with the heat kernel was proposed by \cite{Li:2014}, who studied the case of a 1-sparse source vector.

We begin with a discussion of known results for BPD  from the compressed sensing literature. 
While the theoretical results for BPD do not hold in any meaningful way for ill-conditioned diffusion matrices, we present them here to provide context for the use of this algorithm to recover a \emph{sparse} vector. We then proceed to discussing the performance of BPD on private data in some examples: diffusion on the 1D unit interval and diffusion on general graphs.

\subsection{Basis Pursuit Denoising}


Basis Pursuit Denoising minimises the $\ell_1$-norm subject to the constraint that the measurements of the proposed source vector $\hat{f}$ should be close in the $\ell_2$-norm to the noisy sensor measurements. To simplify our discussion, let $\sigma$ be the standard deviation of the noise added to the sensor measurements. The bound $\sigma\sqrt{m}$ in Algorithm \ref{nonprivatealgo} is chosen to ensure $f_0$ is a feasible point with high probability. 

\begin{algorithm}
   \caption{$R$: Basis Pursuit Denoising}
   \label{nonprivatealgo}
\begin{algorithmic}
   \STATE {\bfseries Input:} $M, \sigma>0, \tilde{y}$
   \STATE $\hat{f} = \arg\min_{f\in[0,1]^n}\|f\|_1 \;\; \text{ s.t. } \|Mf-\tilde{y}\|_2\le \sigma\sqrt{m}$
   \STATE {\bfseries Output:} $\hat{f}\in[0,1]^n$
\end{algorithmic}
\end{algorithm}

The bound $\sigma\sqrt{m}$ in Algorithm \ref{nonprivatealgo} is chosen to ensure $f_0$ is a feasible point with high probability.

\begin{lem}\cite{Hsu:2012}\label{l2noise}
Let $\nu\sim N(0, \sigma^2I_m)$ then for all $t>0$, $$\mathbb{P}[\|\nu\|_2^2>\sigma^2(m+2\sqrt{mt}+2t)]\le e^{-t}.$$ So for large $m$ and small $\rho$, we have $\|\nu\|_2\le (1+\rho)\sigma\sqrt{m}$ with high probability. 
\end{lem}


\subsection{Basis Pursuit Denoising for RIP matrices}

In order to present the results in this section cleaner, rather than keeping track of $\sigma\sqrt{m}$ we introduce parameters $\alpha, \beta>0$. Basis Pursuit Denoising 
\begin{equation}\label{l1min}
\arg\min_{f\in[0,1]^n}\|f\|_1 \;\; \text{ s.t. } \|Mf-\tilde{y}\|_2\le \alpha
\end{equation}
is the convex relaxation of the problem we would like to solve, $\ell_0$-minimisation: 
\begin{equation}\label{l0min}
\arg\min_{f\in[0,1]^n}\|f\|_0 \;\; \text{ s.t. } \|Mf-\tilde{y}\|_2\le \beta.
\end{equation}
The minimum of the $\ell_0$ norm is the \emph{sparsest} solution. Unfortunately, this version of the problem is NP hard, so in order to find an efficient algorithm we relax to the $\ell_1$ norm. The $\ell_1$ norm is the ``smallest" convex function that places a unit penalty on unit coefficients and zero penalty on zero coefficients. Since the relaxation is convex, we can use convex optimisation techniques to solve it. In the next section we'll discuss an appropriate optimisation algorithm. In this section, we focus on when the solution to the relaxed version \eqref{l1min} is similar to the solution for Equation~\eqref{l0min}.

We call the columns of $M$, denoted by $M_i$, atoms. We will assume for this section that $\|M_i\|_2=1$ for all $i$. Notice that the vector $Mf_0$ is the linear combination of the $M_i$ with coefficients given by the entries of $f_0$ so we can think of recovering the vector $f_0$ as recovering the coefficients\footnote{This is where BPD gets its name. We are pursuing the basis vectors that make up $Mf_0$.}. A key parameter of the matrix $M$ is its \emph{coherence}:
\[\mu = \max_{i,j}|\langle M_i, M_j\rangle|\]
Similar to $\triangle_2(M)$, the coherence is a measure how similar the atoms of $M$. The larger the coherence is, the more similar the atoms are, which makes them difficult to distinguish. For accurate sparse recovery, it is preferential for the coherence to be small. The following theorem relates the solutions to Equation \eqref{l1min} and \eqref{l0min}.

\begin{thm}\citep[Theorem C]{Tropp:2004} \label{sparserecovery}
Suppose $k\le \frac{1}{3}\mu^{-1}$. Suppose $f_{opt}$ is a $k$-sparse solution to Equation \eqref{l0min} with $\beta=\frac{\alpha}{\sqrt{1+6k}}$. Then the solution $\hat{f}$ produced from Algorithm \ref{nonprivatealgo} satisfies 
\begin{itemize}
\item supp$(\hat{f})\subset$supp$(f_{opt})$
\item $\hat{f}$ is no sparser than a solution to Equation \eqref{l0min} with $\beta = \alpha$
\item $\|\hat{f}-f_{opt}\|_2\le \alpha\sqrt{3/2}$
\end{itemize}
\end{thm}

Theorem \ref{sparserecovery} says that if the matrix $M$ is coherent then the solution to the convex relaxation (Algorithm \ref{nonprivatealgo}) is at least as sparse as a solution to \eqref{l0min} with error tolerance somewhat smaller than $\alpha$. Also, $\hat{f}$ only recovers source locations that also appear in $f_{opt}$, although it may not recover all of the source locations that appear in $f_{opt}$. The final property bounds the weight assigned to any source identified in $f_{opt}$ and not $\hat{f}$. If $\tilde{y}~=~Mf_0+Z$ then the worst case discrepancy between $f_0$ and $f_{opt}$ occurs when $Z$ concentrates its weight on a single atom. In our case, the noise vector $Z$ has i.i.d. Gaussian coordinates and hence is unlikely to concentrate its weight. 


The key property for \emph{exact} recovery of $f_0$, rather than $f_{opt}$, is that $M$ is a near isometry on sparse vectors. A matrix $M$ satisfies the \emph{Restricted Isometry Property (RIP)} of order $k$ with restricted isometry constant $\delta_k$ if $\delta_k$ is the smallest constant such that for all $k$-sparse vectors $x$, \[(1-\delta_k)\|x\|_2^2\le \|Mx\|_2^2\le (1+\delta_k)\|x\|_2^2.\]

If $f_0$ is a feasible point and $\delta_k$ is small, then we can guarantee that $f_0$ and $\hat{f}$ are close in the $\ell_2$ norm. 

\begin{thm}\citep[Theorem 1.2]{Candas:2008}
\label{RIP}
Assume $\delta_{2k}<\sqrt{2}-1$ and $\|Z\|_2\le \alpha$ and $f_0$ is $k$-sparse. Then the solution $\hat{f}$ to \eqref{l1min} satisfies \[\|\hat{f}-f_0\|_2\le C\alpha\] for some constant $C$. 
\end{thm}

The exact constant $C$ is given explicitly in \cite{Candas:2008} and is rather small. For example, when $\delta_{2k}=0.2$, we have $C\le 8.5$.

Theorems \ref{sparserecovery} and \ref{RIP} only provide meaningful results for matrices with small $\mu$ and $\delta_k$. Unfortunately, the coherence and restricted isometry constants for ill-conditioned matrices, and in particular diffusion matrices, are both large. It is somewhat surprising then that BPD recovers well in the examples we will explore in the following sections.

\section{Diffusion on the Unit Interval}\label{heatonline}

\begin{figure}
	\centering
	\includegraphics[width=0.6\linewidth]{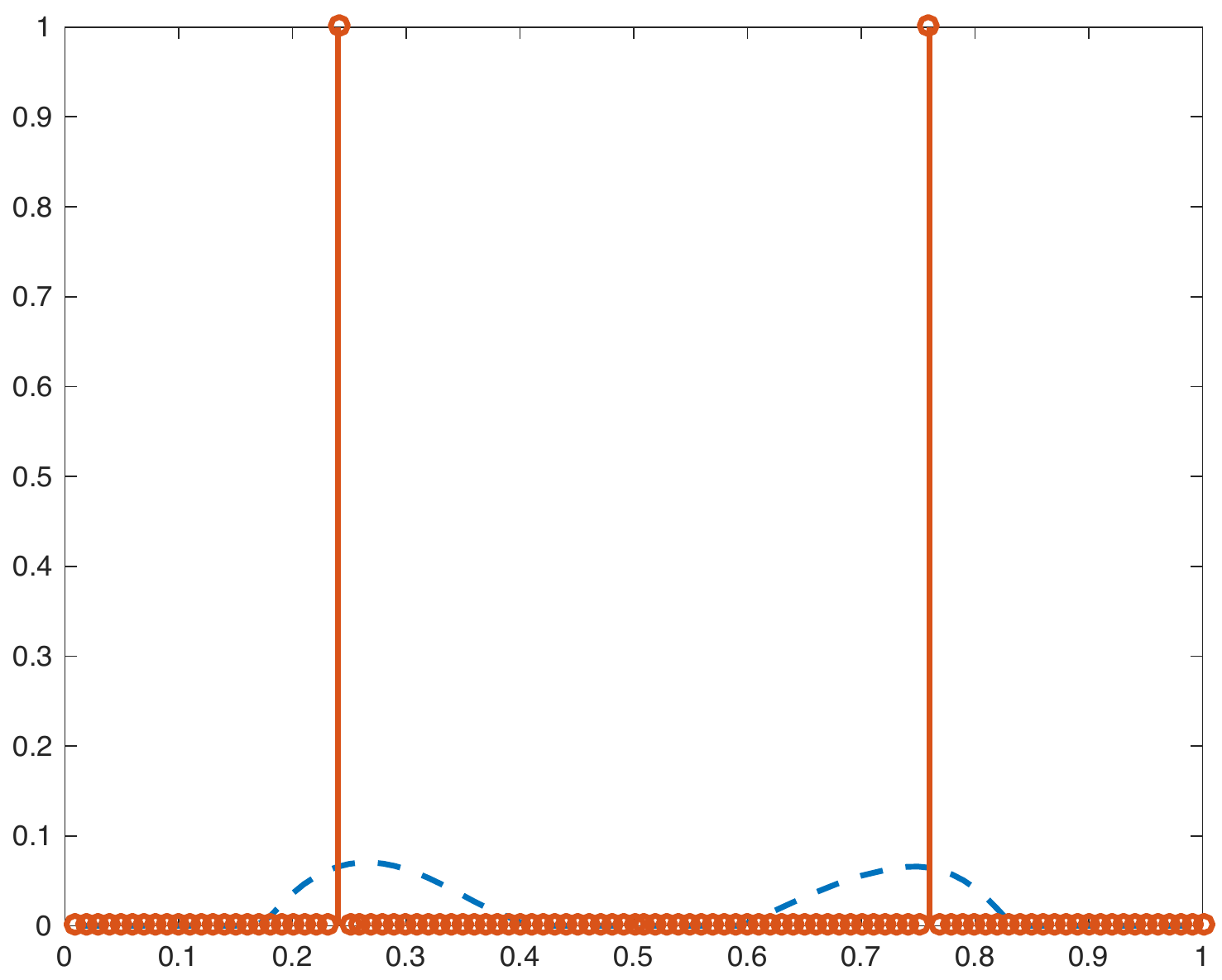}
\caption{$f_0$ has unit peaks 0.76 and 0.24. The parameters are $n=100$, $m=50$, $T=0.05$ and $\sigma = 0.1$. The red line is $f_0$, the blue (dashed) line is $\hat{f}$.}
\label{linediffuse}
\end{figure}

Let us define the linear physical transformation explicitly for heat source localization. To distinguish this special case from the general, we denote the measurement matrix by $A$ (instead of $M$). For heat diffusion, we have a diffusion constant $\mu$ and a time $t$ at which we take our measurements. Let $T = \mu t$ in what follows. Let $g(x,t) = \frac{1}{\sqrt{4\pi T}}e^{\frac{-|x|^2}{4T}}$. Let $n>0$ and suppose the support of $f$ is contained in the discrete set $\{\frac{1}{n}, \cdots, 1\}$. Let $m>0$ and suppose we take $m$ measurements at locations $\frac{1}{m}, \cdots, 1$ so $y_i$ is the measurement of the sensor at location $\frac{i}{m}$ at time $t$ and we have 
\[
	y = Af_0 \;\; \text{ where }\;\; A_{ij} = g\left(\frac{i}{n}-\frac{j}{m}, t\right).
\] 
The heat kernel, $A$, is severely ill-posed due to the fact that as heat dissipates, the measurement vectors for different source vectors become increasingly close \cite{Weber:1981}. Figure \ref{linediffuse} shows the typical behaviour of Algorithm \ref{nonprivatealgo} with the matrix $A$. As can be seen in the figure, this algorithm returns an estimate $\hat{f}$ that is indeed close to $f_0$ in the EMD but not close in more traditional norms like the $\ell_1$ and $\ell_2$ norms. This phenomenon was noticed by \cite{Li:2014}, who proved that if $f_0$ consists of a single source then EMD$(f_0, \hat{f})$ is small where $\hat{f}=R(\tilde{y})$.  

\begin{table}
\caption{Asymptotic upper bounds for private recovery assuming $\sqrt{T}(\frac{\sqrt{m}}{nT^{1.5}}+ke^{\frac{-\alpha^2}{4T}})\le c<1$.}
\label{privatetable}
\vskip 0.15in
\begin{center}
\begin{small}
\begin{sc}
\begin{tabular}{lcccr}
\hline
Variable & \text{EMD}$\left(\frac{f_0}{\|f_0\|_1}, \frac{\hat{f}}{\|\hat{f}\|_1}\right)$ \\
\hline
\hline
$n$&$O(1+\frac{1}{\sqrt{n}})$\\
\hline
$m$&$O(1)$\\
\hline
$T$&$\min\{1, \;O(1+\frac{1}{T}+T^{2.5}e^{-\alpha^2/4T})\}$\\
\hline
\end{tabular}
\end{sc}
\end{small}
\end{center}
\vskip -0.1in
\end{table}


\begin{prop}\label{propstddev}
With the definition of neighbours presented in Definition \ref{neighbours} and restricting to $f_0\in[0,1]^n$ we have
 $$\triangle_2 (A) = O\left(\frac{\alpha\sqrt{m}}{T^{1.5}}\right)$$
\end{prop}

\begin{proof}
For all $i\in[n]$ we have 
\begin{align*}
\|A_i-A_{i+1}\|_2^2 &= \frac{1}{4\pi T}\sum_{j=1}^m \left(e^{\frac{-(\frac{i}{n}-\frac{j}{m})^2}{4T}}-e^{\frac{-(\frac{i+1}{n}-\frac{j}{m})^2}{4T}}\right)^2\\
&= \frac{1}{4\pi T}\sum_{j=1}^m e^{\frac{-(\frac{i}{n}-\frac{j}{m})^2}{2T}}\left(1-e^{\frac{(\frac{i}{n}-\frac{j}{m})^2-(\frac{i+1}{n}-\frac{j}{m})^2}{4T}}\right)^2\\
& \le \frac{1}{4\pi T}\max_{i\in[n]}\max_{j\in[m]}\left(1-e^{\frac{(\frac{i}{n}-\frac{j}{m})^2-(\frac{i+1}{n}-\frac{j}{m})^2}{4T}}\right)^2\sum_{j=1}^m e^{\frac{-(\frac{i}{n}-\frac{j}{m})^2}{2T}}\\
\end{align*}
Now, $\sum_{j=1}^me^{\frac{-(\frac{i}{n}-\frac{j}{m})^2}{2T}}\le m$ and $$\max_{i\in[n]}\max_{j\in[m]}\left(1-e^{\frac{(\frac{i}{n}-\frac{j}{m})^2-(\frac{i+1}{n}-\frac{j}{m})^2}{4T}}\right)^2\le \max\{(1-e^{\frac{-3}{4nT}})^2, (1-e^{\frac{2}{4nT}})^2\}=O\left(\frac{1}{n^2T^2}\right).$$ Therefore, $$\|A_i-A_{i+1}\|_2 = O\left(\frac{\sqrt{m}}{nT^{1.5}}\right).$$
\end{proof}

Figure \ref{stddev} shows calculations of $\triangle_2(A)$ with varying parameters. The vertical axes are scaled to emphasise the asymptotics. These calculations suggest that the analysis in Proposition \ref{propstddev} is asymptotically tight in $m$, $n$ and $T$.

\begin{figure}
\centering
\begin{subfigure}{.33\textwidth}
	\centering
	\includegraphics[width=1\linewidth]{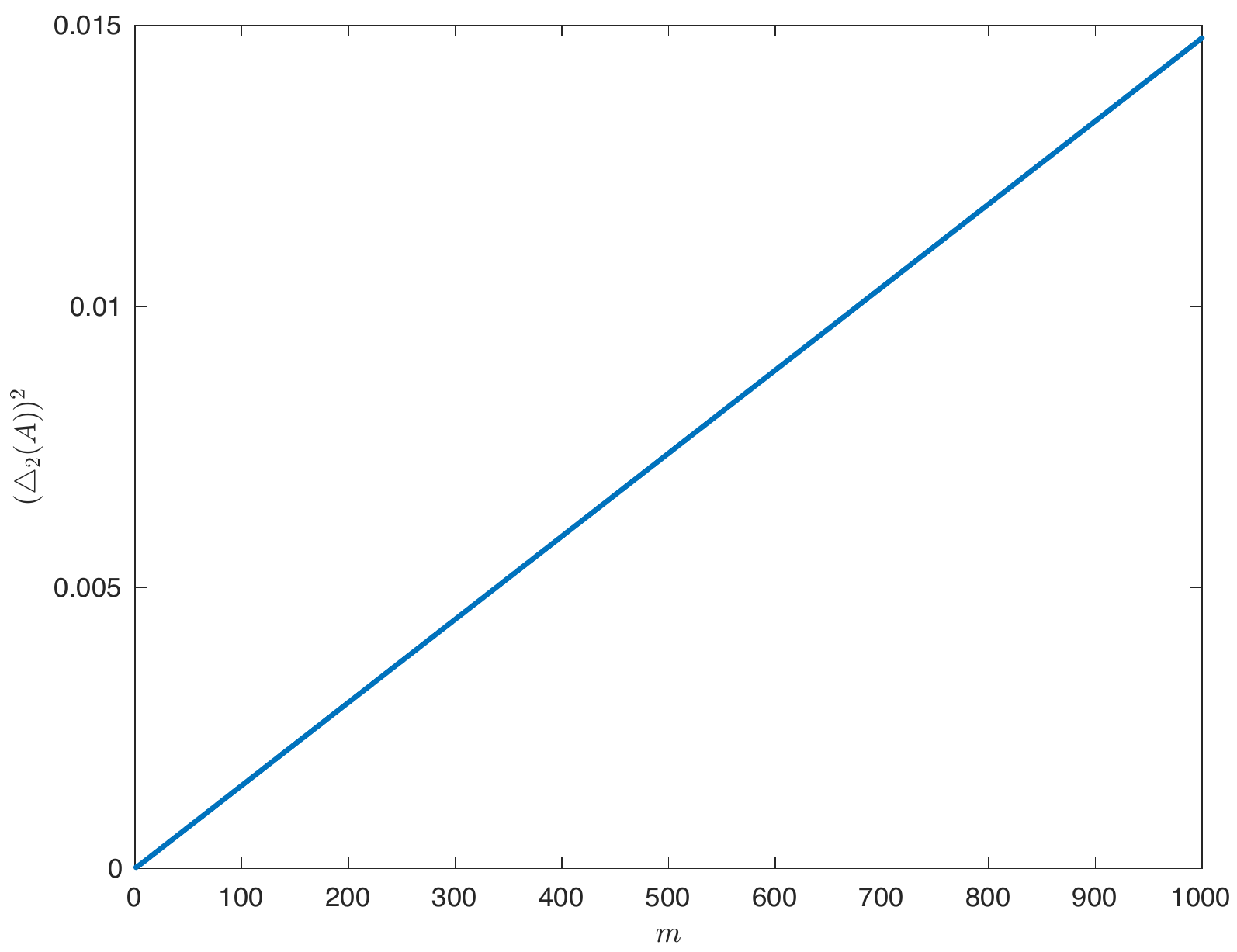}
	\caption{Dependence on $m$}
	\label{mvaries}
\end{subfigure}\hspace{0.01in}%
\begin{subfigure}{.33\textwidth}
	\centering
	\includegraphics[width=1\linewidth]{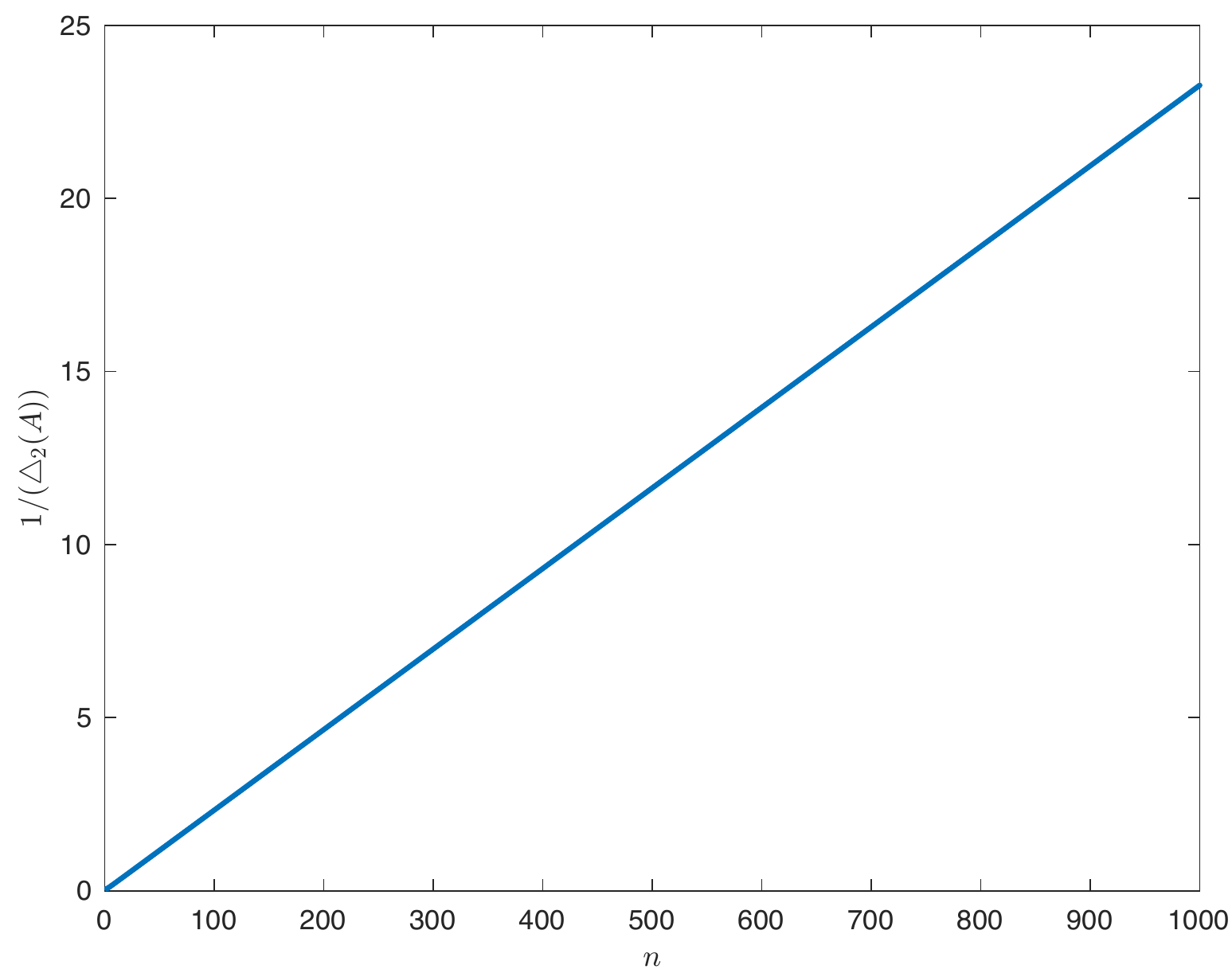}
	\caption{Dependence on $n$}
	\label{nvaries}
\end{subfigure}\hspace{0.01in}%
\begin{subfigure}{.33\textwidth}
	\centering
	\includegraphics[width=1\linewidth]{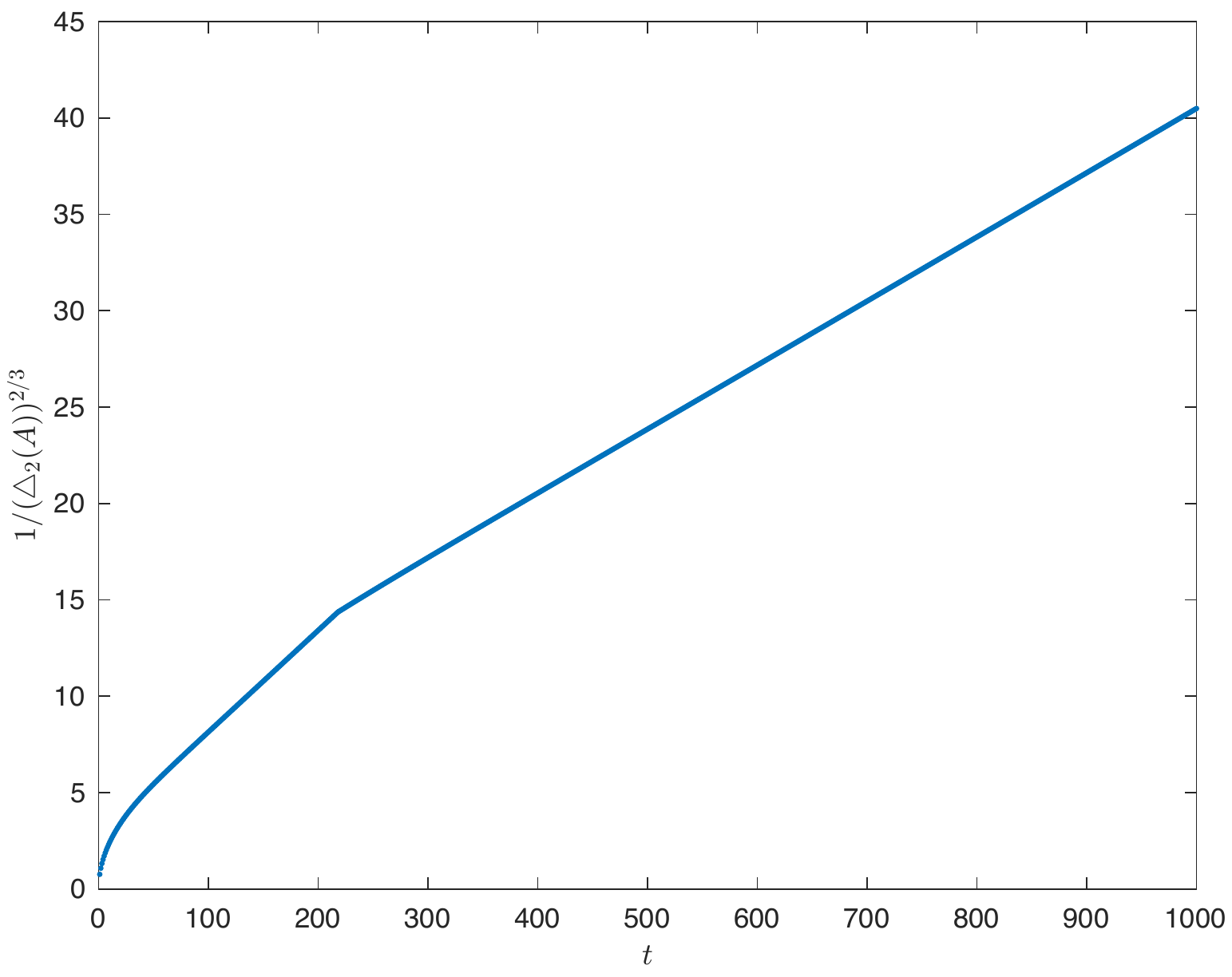}
	\caption{Dependence on $t$}
	\label{tvaries}
\end{subfigure}
\caption{Empirical results of computation of $\triangle_2(A)$. Unless specified otherwise, $m=500$ and $t=0.1$. In \eqref{mvaries}, $n=500$ and in \eqref{tvaries}, $n=1000$.}
\label{stddev}
\end{figure}

\begin{thm} \label{main}
Suppose that $f_0$ is a source vector, $\hat{y}=R(\tilde{y})$ and assume the following:
\begin{enumerate} 
\item\label{enoughsensors} $m\sqrt{T/2}>1$
\item\label{soon} $\sqrt{2T}<1$
\item\label{far} $|x_i-x_j|>\sqrt{2T}+2A$ for some $A>0$
\end{enumerate}
then w.h.p. 
\begin{align*}
\text{EMD}&\Bigg(\frac{f_0}{\|f_0\|_1}, \frac{\hat{f}}{\|\hat{f}\|_1}\Bigg)\le \min\Bigg\{1, \;O\Bigg[\frac{1}{1-\min\{1, C\}}\Bigg(\frac{1}{k}\sqrt{\frac{T^{1.5}C}{\sqrt{T}+1}}+k\min\{1, C\}+\frac{T^2C}{(T+1)k}\Bigg)\Bigg]\Bigg\}
\end{align*}
where $C=\min\left\{k,\; \sqrt{T}\left[\sigma+ke^{-A^2/4T}\right]\right\}$.
\end{thm} 

Assumptions \ref{enoughsensors} and \ref{soon} state that $m$ needs to be large enough that for each possible source and we need to take the measurements before the heat diffuses too much. Assumption \ref{far} says that the sources need to be sufficiently far apart. We can remove this assumption by noting that every source vector is close to a source vector whose sources are well separated and that for all $f, f'$, $\|Af_0~-~Af_0'\|_2~=~O\left(\frac{\sqrt{m}}{ T^{1.5}} \text{EMD}(f_0, f_0')\right)$.  

The result is a generalisation to source vectors with more than one source of a result of \cite{Li:2014} . Our proof is a generalisation of their proof and is contained in Section~\ref{appendixEMD}.
In order to obtain a recovery bound for the private data, we set $\sigma=\left(\frac{\sqrt{m}}{nT^{1.5}}\right)$. The asymptotics of this bound are contained in Table \ref{privatetable}. It is interesting to note that, unlike in the constant $\sigma$ case, the error increases as $T\to 0$ (as well as when $T\to\infty$). This is because as $T\to0$ the inverse problem becomes less ill-conditioned so we need to add more noise. 

The following theorem gives a lower bound on the estimation error of the noisy recovery problem. 

\begin{thm}\label{lowerbound}
We have $$\inf_{\hat{f}}\sup_{f_0}\mathbb{E}[\text{EMD}(f_0, \hat{f})]= \Omega\left(\min\left\{\frac{1}{2}, \frac{T^{1.5}\sigma}{\sqrt{m}}\right\}\right).$$
where $\inf_{\hat{f}}$ is the infimum over all estimators $\hat{f}~:~\mathbb{R}^m~\to~[0,1]^n$, $\sup_{f_0}$ is the supremum over all source vectors in $[0,1]^n$ and $\tilde{y}$ is sampled from $y+N(0, \sigma^2I_m)$.
\end{thm}

Note that this lower bound matches our upper bound asymptotically in $\sigma$ and is slightly loose in $T$. It varies by a factor of $\sqrt{m}$ from our theoretical upper bound. Experimental results (contained in the extended version) suggest that the error decays like $O(1+\frac{1}{\sqrt{m}})$. A consequence of Theorem \ref{lowerbound} is that if two peaks are too close together, roughly at a distance of $O\left(\min\{\frac{1}{2}, \frac{T^{1.5}\sigma}{\sqrt{m}}\right)$, then it is impossible for an estimator to differentiate between the true source vector and the source vector that has a single peak located in the middle. Before we prove Theorem \ref{lowerbound} we need following generalisation of the upper bound in Proposition \ref{propstddev}.

\begin{lem}\label{overlap}
Suppose $\|f_0\|_1=\|f_0'\|_1=1$ then $$\|Af_0-Af_0'\|_2= O\left(\frac{\sqrt{m}}{ T^{1.5}} \text{EMD}(f_0, f_0')\right)$$
\end{lem}

\begin{proof}
Firstly, consider the single peak vectors $e_i$ and $e_j$. Then noting that $Ae_i=A_i$, we have from Proposition \ref{propstddev} that
\begin{align*}
\|Ae_i-Ae_j\|_2 \le\sum_{l=0}^{j-i-1}\|Ae_{i+l}-Ae_{i+1+1}\|_2\le O\left(|i-j|\frac{\sqrt{m}}{nT^{1.5}}\right)
\end{align*}
Now, let $f_{ij}$ be the optimal flow from $f_0$ to $f_0'$ as described in Definition \ref{EMD} so $f_0=\sum_{i,j}f_{ij}e_i$ and $f_0'=\sum_{ij}f_{ij}e_j$. Then 
\begin{align*}
\|Af_0-Af_0'\|_2 &\le \sum_{ij}f_{ij}\|Ae_i-Ae_j\|_2\le O\left(\sum_{ij}f_{ij} \left|\frac{i}{n}-\frac{j}{n}\right|\frac{\sqrt{m}}{T^{1.5}}\right)= O\left(\frac{\sqrt{m}}{T^{1.5}} \text{EMD}(f_0, f_0')\right)
\end{align*}
\end{proof}

The proof of Theorem \ref{lowerbound} will be an application of Fano's inequality, a classic result from information theory. Suppose $p, q$ are probability distributions on the same space. Then the Kullback-Leibler (KL) divergence of $p$ and $q$ is defined by $D(p||q) = \int(\log\frac{dp}{dq})dp$. For a collection $T$ of probability distributions, the KL diameter is defined by $$d_{KL}(T) = \sup_{p,q\in T} D(p||q).$$ If $(\Omega, d)$ is a metric space, $\epsilon>0$ and $T\subset\Omega$, then we define the $\epsilon$-packing number of $T$ to be the largest number of disjoint balls of radius $\epsilon$ that can fit in $T$, denoted by $\mathcal{M}(\epsilon, T, d)$. The following version of Fano's lemma is found in \cite{Yu:1997}.

\begin{lem}[Fano's Inequality.]\label{fano}
Let $(\Omega, d)$ be a metric space and $\{\mathbb{P}_{\theta} \: |\: \theta\in\Omega\}$ be a collection of probability measures. For any totally bounded $T\subset\Omega$ and $\epsilon>0$, 
\begin{equation}
\inf_{\hat{\theta}}\sup_{\theta\in\Omega}\mathbb{P}_{\theta}\left(d^2(\hat{\theta}(X), \theta)\ge \frac{\epsilon^2}{4}\right)\ge 1-\frac{d_{KL}(T)+1}{\log\mathcal{M}(\epsilon, T, d)}
\end{equation}
where the infimum is over all estimators.
\end{lem}

\begin{proof}[Proof of Theorem \ref{lowerbound}]
For any source vector $f_0$, let $P_{f_0}$ be the probability distribution induced on $\mathbb{R}^m$ by the process $Af_0+N(0, \sigma^2I_m)$. Then the inverse problem becomes estimating which distribution $P_{f_0}$ the perturbed measurement vector is sampled from. Let $f_0$ and $f_0'$ be two source vectors. Then 
\begin{align*}
D(P_{f_0}||P_{f_0'}) &= \sum_{i=1}^m \frac{((Af_0)_i-(Af_0')_i)^2}{2\sigma^2}=  \frac{1}{2\sigma^2}\|Af_0-Af_0'\|_2^2\le C\frac{m}{T^3\sigma^2}(\text{EMD}(f_0, f_0'))^2
\end{align*}
for some constant $C$, where we use the fact that the KL-divergence is additive over independent random variables, along with Lemma \ref{overlap}. Now, let $a=\min\{\frac{1}{2}, \frac{T^{1.5}\sigma}{\sqrt{2C}\sqrt{m}}\}$. Let $T$ be the set consisting of the following source vectors: $e_{1/2}$, $(1/2)e_{1/2-a/2}+(1/2)e_{1/2+a/2}$, $(1/4)e_{1/2-a}+(1/2)e_{1/2}+(1/4)e^{1/2+a}$, $(1/2)e_{1/2}+(1/2)e_{1/2+a}$, which are all at an EMD $a$ from each other. Then $d_{KL}(T)+1\le 3/2$ and $\log\mathcal{M}(a, T, \text{EMD})=2$. Thus, by Lemma~\ref{fano}, $$\inf_{\hat{f}}\sup_{f_0}\mathbb{E}[\text{EMD}(f_0, \hat{f})]\ge \frac{3}{4}a = \Omega\left(\min\left\{\frac{1}{2}, \frac{T^{1.5}\sigma}{\sqrt{m}}\right\}\right).$$
\end{proof}


\section{Diffusion on Graphs}
\label{sec:graphdiffusion}

In this section we generalise to diffusion on an arbitrary graph. As usual, our aim is to protect the exact location of a source, while allowing the neighbourhood to be revealed. Diffusion on graphs models not only heat spread in a graph, but also the path of a random walker in a graph and the spread of rumours, viruses or information in a social network. A motivating example is \emph{whisper} networks where participants share information that they would not like attributed to them. We would like people to be able to spread information without fear of retribution, but also be able to approximately locate the source of misinformation. The work in this section does not directly solve this problem since in our setting each node's data corresponds to their \emph{probability} of knowing the rumour, rather than a binary \emph{yes}/\emph{no} variable. In future work, we would like to extend this work to designing whisper network systems with differential privacy guarantees. If a graph displays a community structure, then we would like to determine which community the source is in, without being able to isolate an individual person within that community. 

Let $G$ be a connected, undirected graph with $n$ nodes. The $n\times n$ matrix $W$ contains the edge weights so $W_{ij}=W_{ji}$ is the weight of the edge between node $i$ and node $j$ and the diagonal matrix $D$ has $D_{ii}$ equal to the sum of the $i$-th row of $W$. The graph Laplacian is $L=D-W$. As above, we also have a parameter controlling the rate of diffusion $\tau$. Then if the initial distribution is given by $f_0$ then the distribution after diffusion is given by the linear equation $y=e^{-\tau L}f_0$ \cite{Thanou:2017}. We will use $A_G$ to denote the matrix $e^{-\tau L}$. Note that, unlike in the previous section, we have no heat leaving the domain (i.e., the boundary conditions are different).

The graph $G$ has a metric on the nodes given by the shortest path between any two nodes. Recall that in Lemma \ref{spectrum} we can express the amount of noise needed for privacy, $\Delta_2(A_G)$, in terms of the spectrum of $A_G$. Let $s_1\le s_2\le \cdots\le s_{\min\{n,m\}}$ be the eigenvalues of $L$ then $e^{-\tau s_1}\ge \cdots\ge e^{-\tau s_{\min\{m,n\}}}$ are the eigenvalues of $A_G$.  For any connected graph $G$, the Laplacian $L$ is positive semidefinite and 0 is a eigenvalue with multiplicity 1 and eigenvector the all-ones vector. 

\begin{table}
\caption{Examples of $\Delta_2(A_G)$ for some standard graphs. Each graph has $n$ nodes.}
\label{graphexamples}
\vskip 0.15in
\begin{center}
\begin{small}
\begin{sc}
\begin{tabular}{lccr}
\hline
$G$ & $\Delta_2(A_G)^2$ \\
\hline
\hline
Complete graph & $2e^{-\tau n}$\\
\hline
Star graph & \hspace{-0.35in} $e^{-2\tau n}+\left(\frac{e^{-\tau n}-e^{-\tau}}{n-1}\right)^2+\left(\frac{e^{-\tau n}-e^{-\tau}}{n-1}+e^{-\tau}\right)^2$\\
\hline
\end{tabular}
\end{sc}
\end{small}
\end{center}
\vskip -0.1in
\end{table}

\begin{lemma}\label{rowleftsingular}
For any graph $G$, \[\Delta_2(A_G)~\le~\sum_{k=2}^{\min\{n,m\}}~e^{-\tau s_i}|(U_i)_k-(U_j)_k|,\] where $U_i$ is the $i$th row of the matrix whose columns are the left singular vectors of $L$.
\end{lemma}

\begin{proof}
With set-up as in Lemma \ref{spectrum} we have 
\begin{align*}
\|(A_G)_i-(A_G)_{j}\|_2 &= \|\sum_k e^{-\tau s_k}((U_k)_i-(U_k)_{j})V_k\|_2\le \sum_k e^{-\tau s_k}|(U_k)_i- (U_k)_{j}|
\end{align*}
Since the first eigenvector of $L$ is the all ones vector, we $|(U_1)_i- (U_1)_{j}|=0$.
\end{proof}

An immediate consequence of Lemma \ref{rowleftsingular} is that $\Delta_2(A_G)$ is bounded above by \[e^{-\tau s_2}\|U_i~-~U_j\|_1.\] The second smallest eigenvalue of $L$, $s_2$, (called the \emph{algebraic connectivity}) is related to the connectivity of the graph, in particular the graphs expanding properties, maximum cut, diameter and mean distance \citep{Mohar:1991}. As the graph becomes more connected, the rate of diffusion increases so the amount of noise needed for privacy decreases. The dependence on the rows of the matrix of left singular vectors is intriguing as these rows arise in several other areas of numerical analysis. Their $\ell_2$ norms are called leverage scores \cite{Drineas:2012} and they appear in graph clustering algorithms. 

Figure \ref{graphdiffuse} shows the average behaviour of Algorithm \ref{nonprivatealgo} on a graph with community structure. Preliminary experiments suggest that provided $\tau$ is not too large or too small and $\epsilon$ is not too small, the correct community is recovered.

\begin{figure}
	\centering
	\includegraphics[width=0.5\linewidth]{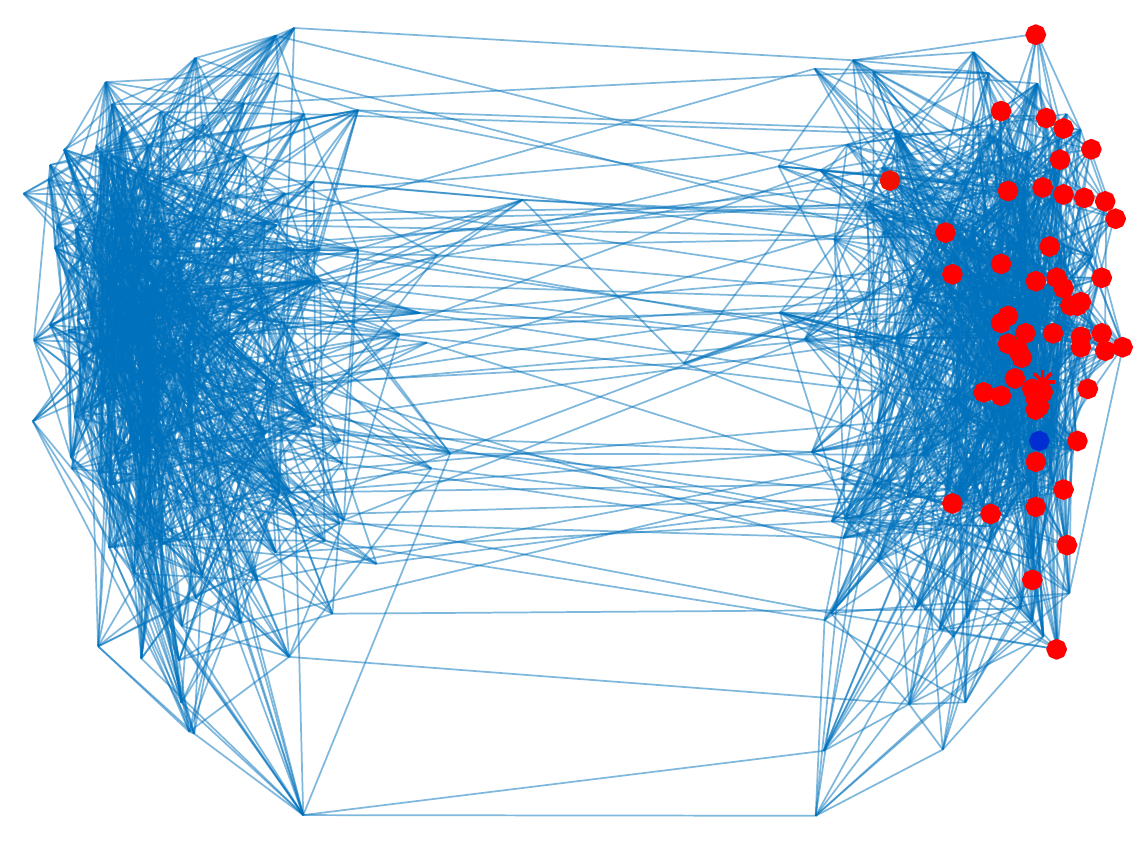}
\caption{$f_0$ has a unit peak at one of the nodes on the left side. The graph $G$ was drawn from a stochastic block model with intercommuntiy probability 5\% and intracommunity probability 0.1\%. The parameters are $\tau = 2, n = 500, \delta = 0.1, \epsilon = 4$.}
\label{graphdiffuse}
\end{figure}

\bibliography{heatprivacybib,bibliography_heat}

\begin{thebibliography}{37}
\providecommand{\natexlab}[1]{#1}
\providecommand{\url}[1]{\texttt{#1}}
\expandafter\ifx\csname urlstyle\endcsname\relax
  \providecommand{\doi}[1]{doi: #1}\else
  \providecommand{\doi}{doi: \begingroup \urlstyle{rm}\Url}\fi

\bibitem[Beddiaf et~al.(2015)Beddiaf, Autrique, Perez, and Jolly]{Beddiaf:2015}
Sara Beddiaf, Laurent Autrique, Laetitia Perez, and Jean-Claude Jolly.
\newblock Heating source localization in a reduced time.
\newblock 26, 02 2015.

\bibitem[Belsley et~al.(1980)Belsley, Kuh, and Welsch]{Belsley:1980}
David~A. Belsley, Edwin Kuh, and Roy~E. Welsch.
\newblock \emph{Regression Diagnostics: Identifying influential data and
  sources of collinearity}.
\newblock John Wiley \& Sons, New York, Chichester, 1980.

\bibitem[Beltran et~al.(2013)Beltran, Erickson, and Cerpa]{Beltran:2013}
Alex Beltran, Varick~L. Erickson, and Alberto~E. Cerpa.
\newblock Thermosense: Occupancy thermal based sensing for hvac control.
\newblock In \emph{Proceedings of the 5th ACM Workshop on Embedded Systems For
  Energy-Efficient Buildings}, BuildSys'13, pages 11:1--11:8, New York, NY,
  USA, 2013. ACM.

\bibitem[Bernstein and Fernandex-Granda(2017)]{Bernstein:2017}
Brett Bernstein and Carlos Fernandex-Granda.
\newblock Deconvolution of {P}oint {S}ources: {A} {S}ampling {T}heorem and
  {R}obustness {G}uarantee.
\newblock \emph{arXiv:1707.00808}, July 2017.

\bibitem[Bun et~al.(2014)Bun, Ullman, and Vadhan]{Bun:2014}
Mark Bun, Jonathan Ullman, and Salil Vadhan.
\newblock Fingerprinting codes and the price of approximate differential
  privacy.
\newblock In \emph{Proceedings of the Forty-sixth Annual ACM Symposium on
  Theory of Computing}, STOC '14, pages 1--10, New York, NY, USA, 2014. ACM.

\bibitem[Burger et~al.(2010)Burger, Landa, Tanushev, and Tsai]{Burger:2010}
Martin Burger, Yanina Landa, Nicolay~M. Tanushev, and Richard Tsai.
\newblock \emph{Discovering a Point Source in Unknown Environments}, pages
  663--678.
\newblock Springer Berlin Heidelberg, Berlin, Heidelberg, 2010.

\bibitem[Cand\`es(2008)]{Candas:2008}
Emmanuel~J. Cand\`es.
\newblock The restricted isometry property and its implications for compressed
  sensing.
\newblock \emph{Comptes Rendus Mathematique}, 346\penalty0 (9):\penalty0 589 --
  592, 2008.

\bibitem[Drineas et~al.(2012)Drineas, Magdon-Ismail, Mahoney, and
  Woodruff]{Drineas:2012}
Petros Drineas, Malik Magdon-Ismail, Michael~W. Mahoney, and David~P. Woodruff.
\newblock Fast approximation of matrix coherence and statistical leverage.
\newblock \emph{J. Mach. Learn. Res.}, 13\penalty0 (1):\penalty0 3475--3506,
  December 2012.

\bibitem[Dwork and Roth(2014)]{Dwork:2014}
Cynthia Dwork and Aaron Roth.
\newblock The algorithmic foundations of differential privacy.
\newblock \emph{Found. Trends Theor. Comput. Sci.}, 9:\penalty0 211--407,
  August 2014.

\bibitem[Dwork et~al.(2006)Dwork, McSherry, Nissim, and Smith]{Dwork:2006}
Cynthia Dwork, Frank McSherry, Kobbi Nissim, and Adam Smith.
\newblock Calibrating noise to sensitivity in private data analysis.
\newblock In \emph{Proceedings of the Third Conference on Theory of
  Cryptography}, TCC'06, pages 265--284, Berlin, Heidelberg, 2006.
  Springer-Verlag.

\bibitem[Dwork et~al.(2007)Dwork, McSherry, and Talwar]{Dwork:2007}
Cynthia Dwork, Frank McSherry, and Kunal Talwar.
\newblock The price of privacy and the limits of lp decoding.
\newblock In \emph{Proceedings of the Thirty-ninth Annual ACM Symposium on
  Theory of Computing}, STOC '07, pages 85--94, New York, NY, USA, 2007. ACM.

\bibitem[Eibl and Engel(2016)]{Eibl:2016}
G{\"u}nther Eibl and Dominik Engel.
\newblock Differential privacy for real smart metering data.
\newblock \emph{Computer Science - Research and Development}, pages 1--10,
  2016.

\bibitem[Farmer et~al.(2013)Farmer, Hall, and Esedoḡlu]{Farmer:2013}
Brittan Farmer, Cassandra Hall, and Selim Esedoḡlu.
\newblock Source identification from line integral measurements and simple
  atmospheric models.
\newblock \emph{Inverse Problems and Imaging}, 7\penalty0 (2):\penalty0
  471--490, 2013.

\bibitem[FTC(2015)]{FTC:2015}
FTC.
\newblock Ftc staff report: {I}nternet of {T}hings: {P}rivacy \& {S}ecurity in
  a {C}onnected {W}orld.
\newblock Technical report, Federal Trade Commission, January 2015.

\bibitem[Geng and Viswanath(2016)]{Geng:2016}
Q.~Geng and P.~Viswanath.
\newblock Optimal noise adding mechanisms for approximate differential privacy.
\newblock \emph{IEEE Transactions on Information Theory}, 62\penalty0
  (2):\penalty0 952--969, Feb 2016.

\bibitem[Greenleaf et~al.(2009)Greenleaf, Kurylev, Lassas, and
  Uhlmann]{Greenleaf2009}
Allan Greenleaf, Yaroslav Kurylev, Matti Lassas, and Gunther Uhlmann.
\newblock Cloaking devices, electromagnetic wormholes, and transformation
  optics.
\newblock \emph{SIAM Review}, 51\penalty0 (1):\penalty0 3--33, 2009.

\bibitem[Haber(2008)]{Haber:2008}
E.~Haber.
\newblock Numerical methods for optimal experimental design of large-scale
  ill-posed problems.
\newblock \emph{Inverse Problems}, 24, 2008.

\bibitem[Hsu et~al.(2012)Hsu, Kakade, and Zhang]{Hsu:2012}
Daniel Hsu, Sham Kakade, and Tong Zhang.
\newblock A tail inequality for quadratic forms of subgaussian random vectors.
\newblock \emph{Electronic Communications in Probability}, 17, 2012.

\bibitem[Jelasity and Birman(2014)]{Jelasity:2014}
M\'{a}rk Jelasity and Kenneth~P. Birman.
\newblock Distributional differential privacy for large-scale smart metering.
\newblock In \emph{Proceedings of the 2Nd ACM Workshop on Information Hiding
  and Multimedia Security}, pages 141--146, New York, NY, USA, 2014. ACM.

\bibitem[Landa et~al.(2011)Landa, Tanushev, and Tsai]{Landa:2011}
Y.~Landa, N.~Tanushev, and R.~Tsai.
\newblock Discovery of point sources in the {H}elmholtz equation posed in
  unknown domains with obstacles.
\newblock \emph{Comm. in Math. Sci.}, 9:\penalty0 903--928, 2011.

\bibitem[Li et~al.(2015)Li, Zhou, and Jiang]{Li:2015}
C.~Li, P.~Zhou, and T.~Jiang.
\newblock Differential privacy and distributed online learning for wireless big
  data.
\newblock In \emph{2015 International Conference on Wireless Communications
  Signal Processing (WCSP)}, pages 1--5, Oct 2015.

\bibitem[Li et~al.(2011)Li, Zhang, Winslett, and Yang]{Li:2011}
Yang~D. Li, Zhenjie Zhang, Marianne Winslett, and Yin Yang.
\newblock Compressive mechanism: Utilizing sparse representation in
  differential privacy.
\newblock In \emph{Proceedings of the 10th Annual ACM Workshop on Privacy in
  the Electronic Society}, WPES '11, pages 177--182, New York, NY, USA, 2011.
  ACM.

\bibitem[Li et~al.(2014)Li, Osher, and Tsai]{Li:2014}
Yingying Li, Stanley Osher, and Richard Tsai.
\newblock Heat source identification based on {L}1 constrained minimization.
\newblock \emph{Inverse Problems and Imaging}, 1\penalty0 (1), 2014.

\bibitem[Lin et~al.(2002)Lin, Federspiel, and Auslander]{Lin:2002}
Craig Lin, Clifford~C. Federspiel, and David~M. Auslander.
\newblock Multi-sensor single-actuator control of hvac systems.
\newblock In \emph{Proc. Int. Conf. Enhanced Building Operations}, Richardson,
  TX, October 2002.

\bibitem[Liu et~al.(2012)Liu, Saroiu, Wolman, and Raj]{Liu:2012}
He~Liu, Stefan Saroiu, Alec Wolman, and Himanshu Raj.
\newblock Software abstractions for trusted sensors.
\newblock In \emph{Proceedings of the 10th International Conference on Mobile
  Systems, Applications, and Services}, MobiSys '12, pages 365--378, New York,
  NY, USA, 2012. ACM.

\bibitem[Mironov et~al.(2009)Mironov, Pandey, Reingold, and
  Vadhan]{Mironov:2009}
Ilya Mironov, Omkant Pandey, Omer Reingold, and Salil Vadhan.
\newblock Computational differential privacy.
\newblock In \emph{Proceedings of the 29th Annual International Cryptology
  Conference on Advances in Cryptology}, CRYPTO '09, pages 126--142, Berlin,
  Heidelberg, 2009. Springer-Verlag.

\bibitem[Mohar(1991)]{Mohar:1991}
Bojan Mohar.
\newblock The laplacian spectrum of graphs.
\newblock In \emph{Graph Theory, Combinatorics, and Applications}, pages
  871--898. Wiley, 1991.

\bibitem[Roozgard et~al.(2016)Roozgard, Barzigar, Verma, and
  Cheng]{Roozgard:2016}
Aminmohammad Roozgard, Nafise Barzigar, Pramode Verma, and Samuel Cheng.
\newblock Genomic data privacy protection using compressed sensing.
\newblock \emph{Trans. Data Privacy}, 9\penalty0 (1):\penalty0 1--13, April
  2016.

\bibitem[Rubner et~al.(2000)Rubner, Tomasi, and Guibas]{Rubner:2000}
Yossi Rubner, Carlo Tomasi, and Leonidas~J. Guibas.
\newblock The earth mover's distance as a metric for image retrieval.
\newblock \emph{International Journal of Computer Vision}, 40\penalty0
  (2):\penalty0 99--121, 2000.

\bibitem[Sarwate and Chaudhuri(2013)]{Sarwate:2013}
A.~D. Sarwate and K.~Chaudhuri.
\newblock Signal processing and machine learning with differential privacy:
  Algorithms and challenges for continuous data.
\newblock \emph{IEEE Signal Processing Magazine}, 30\penalty0 (5):\penalty0
  86--94, Sept 2013.

\bibitem[Thanou et~al.(2017)Thanou, Dong, Kressner, and Frossard]{Thanou:2017}
D.~Thanou, X.~Dong, D.~Kressner, and P.~Frossard.
\newblock Learning heat diffusion graphs.
\newblock \emph{IEEE Transactions on Signal and Information Processing over
  Networks}, 3\penalty0 (3):\penalty0 484--499, Sept 2017.

\bibitem[Tropp(2004)]{Tropp:2004}
Joel Tropp.
\newblock Just relax: Convex programming methods for subset selection and
  sparse approximation.
\newblock 04-04, 01 2004.

\bibitem[Walters et~al.(2007)Walters, Liang, Shi, and Chaudhary]{Walters:2007}
J.P. Walters, Z.~Liang, W.~Shi, and V.~Chaudhary.
\newblock \emph{Wireless Sensor Network Security: A Survey}, page 367.
\newblock CRC Press: Boca Raton, FL, USA, 2007.

\bibitem[Wang et~al.(2016)Wang, Zhang, Yang, Lim, and Ma]{Wang:2016}
L.~Wang, D.~Zhang, D.~Yang, B.~Y. Lim, and X.~Ma.
\newblock Differential location privacy for sparse mobile crowdsensing.
\newblock In \emph{2016 IEEE 16th International Conference on Data Mining
  (ICDM)}, pages 1257--1262, Dec 2016.

\bibitem[Weber(1981)]{Weber:1981}
Charles~F. Weber.
\newblock Analysis and solution of the ill-posed inverse heat conduction
  problem.
\newblock \emph{International Journal of Heat and Mass Transfer}, 24\penalty0
  (11):\penalty0 1783--1792, November 1981.

\bibitem[Yin et~al.(2008)Yin, Osher, Goldfarb, and Darbon]{Yin:2008}
Wotao Yin, Stanley Osher, Donald Goldfarb, and Jerome Darbon.
\newblock Bregman iterative algorithms for $\ell_1$-minimization with
  applications to compressed sensing.
\newblock \emph{SIAM J. Img. Sci.}, 1\penalty0 (1):\penalty0 143--168, March
  2008.

\bibitem[Yu(1997)]{Yu:1997}
Bin Yu.
\newblock {Assouad}, {Fano}, and {Le Cam}.
\newblock In David Pollard, Erik Torgersen, and GraceL. Yang, editors,
  \emph{Festschrift for Lucien Le Cam}, pages 423--435. Springer New York,
  1997.

\end{thebibliography}
\bibliographystyle{plainnat}

\onecolumn

\section{Appendix}\label{appendixEMD}

We need some machinery before we can prove Theorem \ref{main}. The following lemma is from \cite{Li:2014}. Since $T=\mu t$ is fixed we will let $g(x)=g(x,t)$.

\begin{lem}\citep{Li:2014}\label{helper}
Suppose $s_1<x<s_2$ and $|s_1-s_2|\le \sqrt{2T}$ and consider the function $W(z) = -g'(s_2-x)g(z-s_1)-g'(x-s_1)g(s_2-z)$. Then $W(z)$ has a single maximum at $x$ and 
\begin{align*}
&W(x)-W(z)
&\begin{cases}
>W(x)-W(s_2-\sqrt{2T}) &\text{for } z\le s_2-\sqrt{2T}\\
\ge C_1\|z-x\|_2^2 & \text{for } z\in[s_2-\sqrt{2T}, s_1+\sqrt{2T}]\\
> W(x)-W(s_1+\sqrt{2T}) & \text{for } z\ge s_1+\sqrt{2T}
\end{cases}
\end{align*}
where $C_1=\inf_{z\in[s_2-\sqrt{2T}, s_1+\sqrt{2T}]}[-W''(z)/2]>0$. 
\end{lem}

The following two lemmas are necessary for our proof of Theorem \ref{main}. For all $i\in [k]$ and $j\in [p]$, let $W_{ij}(z) = -g'(s_{i_{j+p}}-x_i)g(z-s_{i_j})-g'(x_i-s_{i_j})g(s_{i_{j+p}}-z)$. Let $p=m\sqrt{T/2}$. We will often replace the distance between $s_{i_j}$ and $s_{i_{j+p}}$ with $\sqrt{T/2}$ since it is asymptotically equal to the true distance $p/m=\floor{m\sqrt{T/}}/m$ in $m$. 

\begin{lem}\label{helper3} Using the assumptions of Theorem \ref{main} we have
$$\sum_{j=1}^p \inf_{z\in[s_{i_{j+p}}-\sqrt{2T}, s_{i_j}+\sqrt{2T}]}[-W_{ij}''(z)/2]\ge \Omega\left(\frac{m\sqrt{T/8}+1}{T^{2.5}}\right)$$
\end{lem}

\begin{proof}
Note first that $s_{i_{j+p}}-x_1=\sqrt{T/2}-(x_i-s_{i_j})$ and $s_{i_{j+p}}-z=\sqrt{T/2}-(z-s_{i_j})$ for any $z\in[s_{i_{j+p}}-\sqrt{2T}, s_{i_j}+\sqrt{2T}]$. Let $z\in[s_{i_{j+p}}-\sqrt{2T}, s_{i_j}+\sqrt{2T}]$ then
\begin{align*}
&-W_{ij}''(z) = g'(s_{i_{j+p}}-x_i)g''(z-s_{i_j})+g'(x_i-s_{i_j})g''(s_{i_{j+p}}-z) \\
&= \frac{1}{16\pi T^3}\Bigg[(s_{i_{j+p}}-x_i)\left(1-\frac{(z-s_{i_j})^2}{4T}\right)e^{\frac{-(s_{i_{j+p}}-x_i)^2-(z-s_{i_j})^2}{4T}}\\
&\hspace{1.5in}+(x_i-s_{i_j})\left(1-\frac{(s_{i_{j+p}}-z)^2}{4T}\right)e^{\frac{-(x_i-s_{i_j})^2-(s_{i_{j+p}}-z)^2}{4T}}\Bigg]\\
&\ge \frac{1}{2T\sqrt{4\pi T}}\frac{1}{2T\sqrt{4\pi T}}e^{\frac{-5}{8}}\Bigg[(s_{i_{j+p}}-x_i)\left(1-\frac{(z-s_{i_j})^2}{4T}\right)\\
&\hspace{2in}+(x_i-s_{i_j})\left(1-\frac{(s_{i_{j+p}}-z)^2}{4T}\right)\Bigg]\\
&\ge \frac{e^{\frac{-5}{8}}}{16\pi T^3}\min\{(s_{i_{j+p}}-x_i), (x_i-s_{i_j})\}\Bigg(2-\frac{(z-s_{i_j})^2}{4T}-\frac{(s_{i_{j+p}}-z)^2}{4T}\Bigg)\\
&\ge \frac{e^{\frac{-5}{8}}}{16\pi T^3}\min\{(s_{i_{j+p}}-x_i), (x_i-s_{i_j})\}\frac{3}{4}
\end{align*}
Therefore, 
\begin{align*}
\sum_{j=1}^p \inf_{z\in[s_{i_{j+p}}-\sqrt{2T}, s_{i_j}+\sqrt{2T}]}[-W_{ij}''(z)/2] &\ge \frac{3e^{\frac{-5}{8}}}{64\pi T^3}\sum_{j=1}^p\min\{(s_{i_{j+p}}-x_i), (x_i-s_{i_j})\}\\
&=\frac{3e^{\frac{-5}{8}}}{64\pi T^3}2\sum_{i=1}^{p/2}\frac{i}{m}\\
&=\frac{3e^{\frac{-5}{8}}}{64\pi T^3}2\frac{m\sqrt{T/8}(m\sqrt{T/8}+1)}{2m}
\end{align*}
\end{proof}

\begin{lem}\label{helper2} Using the assumptions of Theorem \ref{main} we have 
\begin{align*}
\min_{i\in [k]}&\min_{l:l/n\not\in S_i} \sum_{j=1}^p (W_{ij}(x_i)-W_{ij}(l/n))=\Omega\left(\frac{m\sqrt{T/2}(m\sqrt{T/2}+1)^2}{m^2}\frac{1}{T^{3.5}}\right).
\end{align*}
\end{lem}

\begin{proof}
From Lemma \ref{helper} we know for all $l$ s.t. $l/n\not\in S_i$ we have 
\begin{align*}W_{ij}&(x_i)-W_{ij}(l/n)\ge \min\{W_{ij}(x_i)-W_{ij}(s_{i_{j+p}}-\sqrt{2T}), W_{ij}(x_i)-W_{ij}(s_{i_j}+\sqrt{2T})\}.
\end{align*}

Let's start with 
\begin{align*}
W_{ij}&(x_i)-W_{ij}(s_{i_{j+p}}-\sqrt{2T}) \\
&= -g'(s_{i_{j+p}}-x_i)\left(g(x_i-s_{i_j})-g(s_{i_{j+p}}-\sqrt{2T}-s_{i_j})\right)\\
&\hspace{2in}-g'(x_i-s_{i_j})\left(g(s_{i_{j+p}}-x_i)-g(\sqrt{2T})\right).
\end{align*}

Now, $g(z)$ is concave down for $z\in[-\sqrt{2T}, \sqrt{2T}]$ and $\lambda$-strongly concave on the interval $[-\sqrt{T/2}, \sqrt{T/2}]$ with $\lambda=\frac{-7}{16\sqrt{4\pi}T^{1.5}}e^{\frac{-1}{8}}$ so 
\[
g(x)-g(y)
\begin{cases}
\ge-g'(x)(y-x) & \text{for } x,z\in[-\sqrt{2T}, \sqrt{2T}]\\
\ge-g'(x)(y-x)+\frac{7}{32\sqrt{4\pi}T^{1.5}}e^{\frac{-1}{8}}(y-x)^2  & \text{for } x,z\in[-\sqrt{T/2}, \sqrt{T/2}]
\end{cases}
\]
Thus, since $|s_{i_{j+p}}-s_{i_j}|=p/m\sim\sqrt{T/2}$ and $|x_i-s_{i_j}|\le \sqrt{T/2}$ and $|s_{i_{j+p}}-x_i|\le \sqrt{T/2}$ we have 
\begin{align*}
W_{ij}(x_i)&-W_{ij}(s_{i_{j+p}}-\sqrt{2T})\\ 
&\ge -g'(s_{i_{j+p}}-x_i)\Bigg(-g'(x_i-s_{i_j})(s_{i_{j+p}}-x_i-\sqrt{2T})\\
&\hspace{1.3in}+\frac{7}{32\sqrt{4\pi}T^{1.5}}e^{\frac{-1}{8}}(s_{i_{j+p}}-x_i-\sqrt{2T})^2\Bigg)\\
&\hspace{1.8in}-g'(x_i-s_{i_j})\left(-g'(s_{i_{j+p}}-x_i)(\sqrt{2T}-s_{i_{j+p}}-x_i)\right)\\
&= g'(s_{i_{j+p}}-x_i)g'(x_i-s_{i_j})\frac{7}{32\sqrt{4\pi}T^{1.5}}e^{\frac{-1}{8}}(s_{i_{j+p}}-x_i-\sqrt{2T})^2\\
&\ge (s_{i_{j+p}}-x_i)(x_i-s_{i_j})e^{\frac{-(s_{i_{j+p}}-x_i)^2-(x_i-s_{i_j})^2}{4T}}\frac{7}{512\sqrt{4\pi}T^{3.5}}e^{\frac{-1}{8}}\\
&\ge (s_{i_{j+p}}-x_i)(x_i-s_{i_j})e^{\frac{-1}{8}}\frac{7}{512\sqrt{4\pi}T^{3.5}}e^{\frac{-1}{8}}\\
\end{align*}
where the last inequality follow since $0\le x_i-s_{i_j} = \sqrt{T/2}-(s_{i_{j+p}}-x_i)\le\sqrt{T/2}$. Now, 
\begin{align*}
\sum_{j=1}^p (s_{i_{j+p}}-x_i)(x_i-s_{i_j}) &\ge \sum_{j=1}^p \frac{j}{m}(\sqrt{T/2}-\frac{j}{m})\\
&= \frac{p(p+1)(3m\sqrt{T/2}-2p-1)}{6m^2}\\
&\sim \frac{m\sqrt{T/2}(m\sqrt{T/2}+1)^2}{6m^2}
\end{align*}
\end{proof}

We can now turn to our proof of the upper bound on the EMD error of Algorithm \ref{nonprivatealgo}.

\begin{proof}[Proof of Theorem \ref{main}]
Note that from Lemma \ref{l2noise}, w.h.p.~$f_0$ is a feasible point so $\|\hat{f}\|_1~\le~k$. Let $S_i=(x_i-\sqrt{T/2}, x_i+\sqrt{T/2})\cap[0,1]$ for $i\in[k]$. Then we have that the $S_i$'s are disjoint and each interval $S_i$ contains $\sqrt{2T}m$ sensors. Let $p=\floor{\sqrt{T/2}m}$ and let $s_{i_{1}}<\cdots<s_{i_{p}}$ be the locations of the sensors in $S_i$ to the left of $x_i$ and $s_{i_{p+1}}>\cdots>s_{i_{2p}}$ be the locations to the right. By Condition \ref{far} we know that for any pair $l\in T_i$ and $s_{c_j}$ where $i\neq c$ we have $|l/n-s_{c_j}|\ge A$. 

For all $i\in [k]$ and $j\in [p]$, let \[W_{ij}(z) = -g'(s_{i_{j+p}}-x_i)g(z-s_{i_j})-g'(x_i-s_{i_j})g(s_{i_{j+p}}-z).\]
Let $T_i$ be the set of all $l\in[n]$ such that $l/n\in (x_i-\frac{|x_i-x_{i-1}|}{2})\cap[0,1]$ then $$g(x_i-s_{i_j})-\sum_{l\in T_i}\hat{f}_lg(l/n-s_{i_j})\le y_{i_j}-\hat{y}_{i_j}+\frac{\|\hat{x}\|_1}{\sqrt{4\pi T}}e^{\frac{-A^2}{4T}}.$$ Therefore, 
\begin{align*}
\sum_{i=1}^k&\sum_{j=1}^p \left|W_{ij}(x_i)-\sum_{l\in T_i} \hat{f}_lW_{ij}(l/n)\right| \\
&= \sum_{i=1}^k\sum_{j=1}^p \Bigg[-g'(s_{j+p}-x_i)\left|g(x_i-s_{i_j})-\sum_{l\in T_i}\hat{f}_lg(l/n-s_{i_j})\right|\\
&\hspace{2in}-g'(x_i-s_{i_j})\left|g(s_{i_{j+p}}-x_i)-\sum_{l\in T_i}\hat{f}_lg(s_{i_j}-l/n)\right|\Bigg]\\
&\le C_2\left[\|y-\hat{y}\|_1+\frac{2p\|\hat{x}\|_1}{\sqrt{4\pi T}}e^{\frac{-A^2}{4T}}\right] \\
&\le C_2\left[\sqrt{m}\|y-\hat{y}\|_2+\frac{2p\|\hat{x}\|_1}{\sqrt{4\pi T}}e^{\frac{-A^2}{4T}}\right] \\ 
&\le C_2\left[2\sigma m+\frac{mk}{\sqrt{2\pi}}e^{\frac{-A^2}{4T}}\right]
\end{align*}
where $C_2=\max_{i\in[k]}\max_{j\in[2p]}[-g'(|s_{i_j}-x_i|)]$ and the last inequality holds with high probability from Lemma \ref{l2noise}. Also, 
\begin{align*}
W_{ij}(x_i) &\le -g'(s_{i_{j+p}}-x_i)g(x_i-s_{i_j})-g'(x_i-s_{i_j})g(s_{i_{j+p}}-x_i)\\
&\le \frac{1}{16\pi T^2}|s_{i_{j+p}}-x_i|+\frac{1}{16\pi T^2}|x_i-s_{i_j}|\\
&\le \frac{1}{8\pi T^{1.5}}
\end{align*}
Therefore, since $\sum_{\ell\in T_i}\hat{f_{\ell}}\le 1$ we have
\[\sum_{i=1}^k\sum_{j=1}^p \left|W_{ij}(x_i)-\sum_{l\in T_i} \hat{f}_lW_{ij}(l/n)\right|\le \min\left\{\frac{km}{8\pi T}\;\;,\;\; C_2\left[2\sigma m+\frac{mk}{\sqrt{2\pi}}e^{\frac{-A^2}{4T}}\right]\right\} = B.\]

Conversely, by Lemma \ref{helper}, we have
\begin{align*}
&\sum_{i=1}^k\sum_{j=1}^p \left|W_{ij}(x_i)-\sum_{l\in T_i} \hat{f}_lW_{ij}(l/n)\right| \\
&\ge \sum_{j=1}^p\left[\sum_{i=1}^k \left(1-\sum_{l\in T_i}\hat{f}_l\right) W_{ij}(x_i)+\sum_{i=1}^k\sum_{l\in T_i}\hat{f}_l(W_{ij}(x_i)-W_{ij}(l/n))\right]\\
&\ge \sum_{j=1}^p\sum_{i=1}^k \left(1-\sum_{l\in T_i}\hat{f}_l\right)W_{ij}(x_i)+\sum_{j=1}^p\sum_{i=1}^k\sum_{l:l/n\in S_i}\hat{f}_l C_1^j(x_i-l/n)^2 + \sum_{i=1}^k \sum_{l: l/n\not\in S_i} C_3\hat{f}_l\\
&\ge \sum_{j=1}^p\sum_{i=1}^k \left(1-\sum_{l\in T_i}\hat{f}_l\right)W_{ij}(x_i)+C_5\sum_{i=1}^k\sum_{l:l/n\in S_i}\hat{f}_l (x_i-l/n)^2 + \sum_{i=1}^k \sum_{l: l/n\not\in S_i} C_3\hat{f}_l
\end{align*}
where $C_3=\min_{i\in [k]}\min_{l:l/n\not\in S_i} \sum_{j=1}^p (W_{ij}(x_i)-W_{ij}(l/n))\ge 0.$ 

Now, by the uniformity of the sensor locations, $W_j = W_{i_j}(x_i)=W_{i'_j}(x_{i'})$ so \[\sum_{j=1}^p\sum_{i=1}^k \left(1-\sum_{l\in T_i}\hat{f}_l\right)W_{ij}(x_i)= \sum_{j=1}^p\sum_{i=1}^k \left(1-\sum_{l\in T_i}\hat{f}_l\right)W_j\ge \sum_{j=1}^p W_j(k-\|\hat{f}\|_1)\ge 0.\] Similarly the other two terms are both positive. Therefore, $\sum_{i=1}^k\sum_{l:l/n\not\in S_i}\hat{f}_l\le B/C_3$ or equivalently, $$\sum_{i=1}^k\sum_{l: l/n\in S_i}\hat{f}_l\ge \|\hat{f}\|_1-\min\{k, B/C_3\}  .$$
This implies that most of the weight of the estimate $\hat{f}$ is contained in the intervals $S_1, \cdots, S_k$. Also, 
\begin{align*}
B&\ge \sum_{j=1}^p|W_{ij}(x_i)-\sum_{i\in T_i}\hat{f}_lW_{ij}(l/n)|\\
& \ge \sum_{j=1}^pW_{ij}(x_i)-\sum_{i\in T_i}\hat{f}_lW_{ij}(x_i)\\
&\ge \left(\sum_{j=1}^p W_{ij}(x_i)\right)\left(1-\sum_{l\in T_i}\hat{f}_l\right) = C_4\left(1-\sum_{l\in T_i}\hat{f}_l\right)
\end{align*}
Therefore, $\sum_{l\in T_i}\hat{f}_l\ge 1-\min\{1, B/C_4\}$ and 
\begin{align*}
\sum_{l: l/n\in S_i}\hat{f}_l&\le \sum_{l\in T_i}\hat{f}_l = \|\hat{f}\|_1-\sum_{a\neq i}\sum_{l\in T_a}\hat{f}_l \\
&\le \|\hat{f}\|_1-(k-1)(1-\min\{1, B/C_4\})\le 1+(k-1)\min\{1, B/C_4\}.
\end{align*}
This implies that the weight of estimate $\hat{f}$ contained in the interval $S_i$ is not too much larger than the true weight of 1. Also, $$\frac{\sum_{l:l/n \in S_i}\hat{f}_l}{\|\hat{f}\|_1}\le \frac{1+(k-1)\min\{1, B/C_4\}}{k(1-\min\{1, B/C_4\})}=\frac{1}{k}+\frac{\min\{1, B/C_4\}}{1-\min\{1, B/C_4\}}$$

In order to upper bound the EMD$(\frac{f_0}{k}, \frac{\hat{f}}{\|\hat{f}\|_1})$ we need a flow, we are going to assign weight \[\min\{\frac{\sum_{l:l/n \in S_i}\hat{f}_l}{\|\hat{f}\|_1}, \frac{1}{k}\}\] to travel to $x_i$ from within $S_i$. The remaining unassigned weight is at most \[k\frac{\min\{1, B/C_4\}}{1-\min\{1, B/C_4\}}+\frac{\min\{1, B/C_3\}}{k(1-\min\{1, B/C_4\})}\] and this weight can travel at most 1 unit in any flow. Therefore, 
\begin{align}
\text{EMD}\left(\frac{f_0}{k}, \frac{\hat{f}}{\|\hat{f}\|_1}\right)&\le \sum_{i=1}^k\sum_{l:l/n\in S_i}\frac{\hat{f}_l}{\|\hat{f}\|_1}|x_i-l/n|+k\frac{\min\{1, B/C_4\}}{1-\min\{1, B/C_4\}}+\frac{\min\{1, B/C_3\}}{k(1-\min\{1, B/C_4\})}\nonumber\\
&\le \frac{1}{k(1-\min\{1, B/C_4\})}\sqrt{\frac{B}{C_5}}+k\frac{\min\{1, B/C_4\}}{1-\min\{1, B/C_4\}}+\frac{\min\{1, B/C_3\}}{k(1-\min\{1, B/C_4\})} \label{prelimbound}
\end{align}

Now, we need bounds on $C_1$, $C_2$, $C_3$ and $C_4$. Firstly, recall \[C_1=\inf_{z\in[s_2-\sqrt{2T}, s_1+\sqrt{2T}]}[-W''(z)/2]>0.\] The sensors $s_{i_j}$ and $s_{i_{j+p}}$ are at a distance of $p/m$ and recall that we only chose the sensors such that $|x_i-s_{i_j}|\ge 1/m$. Thus, any $z\in[s_2-\sqrt{2T}, s_1+\sqrt{2T}]$ we have either $|z~-~s_{i_j}|~\le~\sqrt{T/8}$ or $|z-s_{i_{j+p}}|\le \sqrt{T/8}$ so 
\begin{align*}
-W''(z) &= g'(s_{i_{j+p}}-x_i)g''(z-s_{i_j})+g'(x_i-s_{i_j})g''(s_{i_{j+p}}-z) \ge \frac{1}{16\pi T^3}\left(\frac{1}{m}\left(1-\frac{1}{2T}\frac{T}{8}\right)e^{-\frac{1}{4T}\frac{T}{8}}\right)
\end{align*}
Therefore, $C_1\ge \frac{17e^{\frac{-1}{32}}}{512}\frac{1}{m}\frac{1}{T^3}.$ Next, \[C_2 = \max_{i\in[k]}\max_{j\in[2p]} -g'(|x_i-s_j|)\le \frac{1}{2\sqrt{4\pi T}}e^{\frac{-1}{4m^2T}}.\] By Lemma \ref{helper2} we have, \[C_3=\min_{i\in [k]}\min_{l:l/n\not\in S_i} \sum_{j=1}^p (W_{ij}(x_i)-W_{ij}(l/n))=\Omega\left(\frac{m\sqrt{T/2}(m\sqrt{T/2}+1)^2}{m^2}\frac{1}{T^{3.5}}\right).\] Finally, 
\begin{align*}
C_4 &= \sum_{j=1}^p W_{ij}(x_i)\\
&= \frac{1}{8\pi T^2}\sum_{j=1}^p (s_{i_{j+p}}-x_i)e^{\frac{-(s_{i_{j+p}}-x_i)^2-(x_i-s_{i_{j}})^2}{4T}}+(x_i-s_{i_j})e^{\frac{-(s_{i_{j+p}}-x_i)^2-(x_i-s_{i_{j}})^2}{4T}}\\
&\ge \frac{1}{8\pi T^2}e^{\frac{-1}{8}}\sum_{j=1}^p(s_{i_{j+p}}-s_{i_j})\\
&\ge \Omega\left(\frac{me^{\frac{-1}{8}}}{16\pi T}\right).
\end{align*}
Lemma \ref{helper3} gives $C_5=\Omega\left(\frac{m\sqrt{T/8}+1}{T^{2.5}}\right)$. Putting all our bounds into \eqref{prelimbound} we gain the final result.
\end{proof}

\end{document}